\documentclass[lettersize,journal]{IEEEtran}
\usepackage[margin=0.64in]{geometry}

\usepackage[T1]{fontenc}
\usepackage{amssymb}
\usepackage[utf8]{inputenc}
\usepackage[english]{babel}
\usepackage{amsmath}			
\usepackage{courier}

\usepackage{comment}				
\usepackage[T1]{fontenc}
\usepackage{mathtools}
\usepackage{algcompatible,amsmath}
\usepackage[ruled,vlined]{algorithm2e}
\usepackage{graphicx}
\usepackage{multicol}
\usepackage{afterpage}
\usepackage[noend]{algpseudocode}
\usepackage{color}
\usepackage{etoolbox}
\usepackage{url}
\usepackage{siunitx,booktabs}
\usepackage{placeins}
\usepackage{bbm}
\usepackage{soul}
\usepackage[noadjust]{cite}
\DeclareMathOperator{\tr}{tr}
\usepackage{amsfonts}

\usepackage{accents}

\normalsize
\makeatletter
\def\underbracex#1#2{\mathop{\vtop{\m@th\ialign{##\crcr
				$\hfil\displaystyle{#2}\hfil$\crcr
				\noalign{\kern3\p@\nointerlineskip}%
				#1\crcr\noalign{\kern3\p@}}}}\limits}

\def\upbracefilla{$\m@th \setbox\z@\hbox{$\braceld$}%
	\bracelu\leaders\vrule \@height\ht\z@ \@depth\z@\hfill 
	\kern\p@\vrule \@width\p@\kern\p@\vrule \@width\p@\kern\p@\vrule \@width\p@
	$}

\def\upbracefillb{$\m@th \setbox\z@\hbox{$\braceld$}%
	\vrule \@width\p@\kern\p@\vrule \@width\p@\kern\p@\vrule \@width\p@\kern\p@
	\leaders\vrule \@height\ht\z@ \@depth\z@\hfill\bracerd
	\braceld\leaders\vrule \@height\ht\z@ \@depth\z@\hfill
	\kern\p@\vrule \@width\p@\kern\p@\vrule \@width\p@\kern\p@\vrule \@width\p@
	$}

\def\upbracefillc{$\m@th \setbox\z@\hbox{$\braceld$}%
	\vrule \@width\p@\kern\p@\vrule \@width\p@\kern\p@\vrule \@width\p@\kern\p@
	\leaders\vrule \@height\ht\z@ \@depth\z@\hfill
	\kern\p@\vrule \@width\p@\kern\p@\vrule \@width\p@\kern\p@\vrule \@width\p@
	$}

\def\upbracefilld{$\m@th \setbox\z@\hbox{$\braceld$}%
	\vrule \@width\p@\kern\p@\vrule \@width\p@\kern\p@\vrule \@width\p@\kern\p@
	\leaders\vrule \@height\ht\z@ \@depth\z@\hfill\braceru$}
	
\def\underbracex#1#2{\mathop{\vtop{\m@th\ialign{##\crcr
				$\hfil\displaystyle{#2}\hfil$\crcr
				\noalign{\kern3\p@\nointerlineskip}%
				#1\crcr\noalign{\kern3\p@}}}}\limits}

\def\upbracefilla{$\m@th \setbox\z@\hbox{$\braceld$}%
	\bracelu\leaders\vrule \@height\ht\z@ \@depth\z@\hfill 
	\kern\p@\vrule \@width\p@\kern\p@\vrule \@width\p@\kern\p@\vrule \@width\p@
	$}

\def\upbracefillb{$\m@th \setbox\z@\hbox{$\braceld$}%
	\vrule \@width\p@\kern\p@\vrule \@width\p@\kern\p@\vrule \@width\p@\kern\p@
	\leaders\vrule \@height\ht\z@ \@depth\z@\hfill\bracerd
	\braceld\leaders\vrule \@height\ht\z@ \@depth\z@\hfill
	\kern\p@\vrule \@width\p@\kern\p@\vrule \@width\p@\kern\p@\vrule \@width\p@
	$}

\def\upbracefillbd{$\m@th \setbox\z@\hbox{$\braceld$}%
	\vrule \@width\p@\kern\p@\vrule \@width\p@\kern\p@\vrule \@width\p@\kern\p@
	\bracerd\braceld
	\leaders\vrule \@height\ht\z@ \@depth\z@\hfill\braceru$}

\makeatletter
\patchcmd{\@makecaption}
{\scshape}
{}
{} 
{}
\makeatother
\usepackage{subcaption}
\usepackage{graphicx}
\usepackage{amsthm}

\newtheorem{proposition}{Proposition}

\newtheorem*{proof*}{Proof}

\newcommand{\norm}[1]{\left\lVert#1\right\rVert}

\IEEEoverridecommandlockouts
\makeatletter
\let\origIEEEPARstart\IEEEPARstart
\renewcommand{\IEEEPARstart}[3][1.1]{%
	\def\@IEEEPARstartDROPDEPTH{#1\baselineskip}%
	\origIEEEPARstart{#2}{#3}%
}

\def\endthebibliography{%
	\def\@noitemerr{\@latex@warning{Empty `thebibliography' environment}}%
	\endlist
}
\makeatother

\ifCLASSINFOpdf

\else

\fi

\begin{document}
	
	\title{Symbol-Level Mask-Compliant Hybrid Precoding for Multi-User MIMO-OFDM Systems}

	\author{Navid Reyhanian, Parisa Ramezani, and Emil Bj{\"o}rnson,
		\IEEEmembership{Fellow, IEEE}
			\thanks{N. Reyhanian was with the Department of Electrical Engineering, University of Minnesota, Minneapolis, MN, USA 55455. He is now with Cisco Systems, Milpitas, CA, USA 95035  (email: navid@umn.edu).}
		\thanks{P. Ramezani and E. Bj{\"o}rnson are with the Department of Communication Systems, KTH Royal Institute of Technology, Stockholm, Sweden (email: \{parram,emilbjo\}@kth.se).}\vspace{-0.7cm}
	}
	

	\allowdisplaybreaks

	\maketitle

	\begin{abstract}
Millimeter-wave (mmWave) technology is a crucial enabler for next-generation networks because it offers substantially greater available bandwidth. mmWave multiple-input multiple-output (MIMO) systems cannot rely solely on fully digital precoding due to hardware costs. As a result, hybrid precoding, which combines digital baseband processing with RF precoding, has emerged as a practical solution that balances performance and implementation complexity. As mmWave links typically operate over wideband, frequency-selective channels, orthogonal frequency-division multiplexing (OFDM) is commonly used to mitigate dispersive effects, yet OFDM introduces practical drawbacks, including out-of-band (OOB) emissions from abrupt spectral transitions among subcarriers and additional spectral leakage induced by windowing. Moreover, nonideal phase shifters (PS) in the RF transmit precoder and the user combiner impose inherent implementation limits that result in phase errors. We investigate robust joint digital--RF precoder design for minimizing the downlink sum mean-squared error (MSE) in hybrid multi-user (MU) MIMO--OFDM systems subject to maximum transmit-power, clipping, and OOB spectral-mask constraints. The resulting optimization is nonconvex and challenging to solve. To address this, we develop a minimum mean-squared error (MMSE) based block coordinate descent (BCD) algorithm that alternates between updating the transmitter-side digital--RF precoders and the user-side digital--RF combiners. For each BCD subproblem, we propose computationally efficient and scalable, closed-form solution strategies suitable for practical implementation. Extensive simulations validate the proposed methods and show clear performance improvements over established benchmark schemes.
		
	\end{abstract}
	
	\IEEEpeerreviewmaketitle

	\section{Introduction}
	\label{sec:introduction}
\IEEEPARstart{M}{illimeter} wave (mmWave) communication is a strong candidate for easing the spectrum shortage in future wireless networks \cite{7913599}. Although mmWave links suffer higher path loss and are more easily blocked than sub-6~GHz links, the short wavelength permits dense antenna packing, making massive multiple-input multiple-output (MIMO) practical. With suitable precoding and combining, MIMO systems compensate for harsh propagation via spatial multiplexing and diversity gains. These benefits are especially important in multi-user MIMO (MU-MIMO), where a base station (BS) serves several users on the same time--frequency resources and precoding must additionally control inter-user interference.

However, fully-digital precoding is impractical for large arrays because it requires one radio-frequency (RF) chain per antenna, leading to prohibitive hardware cost and power consumption. A widely considered alternative is the hybrid phase-shifter (PS) based digital--RF architecture, which pairs a high-dimensional analog precoder built from low-complexity PS networks with a lower-dimensional digital precoder \cite{7913599,7397861,6928432}, closely approaching fully-digital performance at much lower complexity \cite{8048606}.

A further challenge is that mmWave systems are typically wideband, requiring orthogonal frequency-division multiplexing (OFDM) to handle frequency selectivity. In hybrid OFDM architectures, the RF precoder is shared across all subcarriers while the digital precoder is optimized per subcarrier \cite{9467491}, introducing a cross-frequency coupling absent in flat-fading designs.

Wideband OFDM creates practical difficulties due to nonlinear distortion and out-of-band (OOB) radiation \cite{7366560}. Since OFDM signals can have large peak amplitudes, amplitude clipping is commonly applied to avoid saturation in digital-to-analog converters (DACs) and power amplifiers (PAs), but this causes in-band distortion and spectral regrowth in adjacent channels \cite{9151131}. The problem is more severe at mmWave because wide bandwidths and high transmit powers needed for coverage make OOB emissions more critical \cite{10048700}. Accordingly, wireless standards enforce strict spectral emission masks and adjacent channel leakage ratio (ACLR) limits \cite{9390405}, which motivates hybrid precoding designs that incorporate clipping and spectral constraints, especially when compliance is based on peak-detected measurements and therefore requires symbol-level spectral constraints in the precoder design.

Existing regulations constrain instantaneous peak emissions, both in-band and out-of-band. In particular, the ultra wideband (UWB) rule specifies a peak-power ceiling of 0\,dBm/50\,MHz~\cite{fcc_uwb_15_517}, and Section~96.41 for citizens broadband radio service (CBRS) requires OOB compliance to be assessed through peak-detected measurements~\cite{fcc_cbrs_96_41}. This is especially relevant here because nonlinear analog components placed after the digital precoder can generate spectral regrowth. Therefore, to ensure regulatory compliance in the transmitted waveform, the OOB spectral mask should be imposed on the emitted spectrum of each OFDM symbol \cite{9214877}.

	\subsection{Related Work}
Hybrid precoding has been extensively studied; see, e.g., \cite{7913599,7397861,9467491,10839437,10972245,10778226}. In a partially-connected architecture, each antenna is linked to a single RF chain, keeping complexity and power consumption low; fully-connected designs link every antenna to all RF chains, offering more flexibility at the cost of higher complexity and insertion losses \cite{7913599}. Practical systems typically adopt the partially-connected architecture.

A key difficulty is the RF precoder/combiner design, since the unit-modulus PS constraints make the problem inherently nonconvex. Common approaches include Riemannian conjugate gradient (RCG) methods \cite{7397861,9467491} and phase search-based techniques \cite{wang2018hybrid}.
	
Spectral precoding given deterministic data offers another practical way to shape the transmit spectrum and limit OOB leakage. For example, mask-compliant designs in \cite{6459499,7485853,9214877} directly enforce a target spectral emission mask, rather than trying to cancel the spectrum at a few isolated frequency points. In contrast, \cite{van2009sculpting} proposes a least-squares notching precoder that suppresses OOB components at selected frequencies by keeping the precoded data vector as close as possible (in the Euclidean sense) to the original one. Relatedly, \cite{10048700} considers a SU design that reduces OOB radiation under peak-to-average ratio (PAPR) constraints by adding a small auxiliary suppression signal per OFDM symbol, leveraging cyclic-prefix (CP) and guard tones while remaining transparent to a conventional receiver. The main drawbacks are the additional transmit power and the CP/guard overhead, which can reduce efficiency and throughput.

In many SU designs, the transmit waveform is formed by applying separate spatial and spectral precoders that are designed largely independently; see, e.g., \cite{van2009sculpting,9390405}. This separation becomes problematic in high-rate MU-MIMO-OFDM systems, where spectral precoding can distort the spatially precoded signals and alter their structure in a way that increases MU interference and lowers spectral efficiency. Motivated by this coupling, \cite{taheri2020joint} studies joint spatial--spectral precoding for MU-MIMO-OFDM and proposes schemes based on zero-forcing (ZF) and maximum-ratio transmission (MRT). ZF can be fragile at low signal-to-noise ratio (SNR) because it may suppress the desired signal excessively, whereas MRT mainly boosts the intended user's SNR and is generally insufficient to control inter-user interference. In contrast, minimum mean-square error (MMSE) precoding \cite{1468466} explicitly trades off noise enhancement and interference suppression, and is therefore often preferred in practical MU-MIMO settings \cite{8048606}.
	
PS impairments are unavoidable in practical mmWave hybrid precoding hardware due by fabrication tolerances, non-ideal RF components, and device aging. These effects introduce random phase and gain mismatches relative to the intended PS settings, which can significantly degrade system performance. A common modeling approach treats the impairments as random perturbations around the ideal PS response, with the phase errors often modeled as Gaussian \cite{7956180,7948784,9069262,9767793,8653302}. Recent work has tackled PS uncertainty in different ways: \cite{9767793} develop a robust hybrid precoding design via alternating optimization with outage-probability constraints, while \cite{9069262} proposes estimating an effective downlink channel that inherently captures PS imperfections and then designing the digital precoder based on this estimate (rather than relying on channel reciprocity), improving robustness and reducing training overhead.
	
While the above issues have been explored to some extent in isolation, a unified treatment is needed to address their coupled impact in practical systems. To the best of our knowledge, existing work does not jointly consider all of these aspects within a single robust, mask-compliant precoding framework for hybrid MU-MIMO-OFDM systems, which is the main focus of this paper.
	
\subsection{Our Contributions}
In this paper, we develop a robust hybrid precoding framework for partially-connected MU-MIMO-OFDM transmitters. The key contributions are as follows:
\begin{itemize}
\item We establish a unified sum-MSE minimization formulation that jointly optimizes the transmitter digital and RF precoders and the users' analog and digital combiners, while explicitly enforcing amplitude clipping and OOB emission-mask constraints pointwise for each realized symbol tuple in a batch of independent symbol realizations. In contrast to many existing approaches that mainly aim to reduce OOB emissions as much as possible, our framework is designed to \emph{guarantee} compliance with arbitrary required spectral limits. The resulting joint design tightly couples the spatial and spectral dimensions and is nonconvex in all variables; moreover, the transmit-side digital variables are optimized pointwise for each realized symbol tuple, whereas the receiver-side analog and digital combiners remain symbol-agnostic and are updated from batch/sample-average statistics. To handle this coupling in a principled and efficient manner, we propose an MMSE-based BCD algorithm that alternates over the variable blocks.
	
	\item The digital transmit-side update leads to a large-scale convex problem. Its feasible set is shaped by the transmit-power budget, the per-antenna spectral-mask constraints, the clipping limits, and the linear consistency relations that connect the frequency- and time-domain variables as well as the per-antenna and RF-chain-domain representations. To address this, we express the transmit-side step using four coupled variable blocks and develop a low-complexity, scalable ADMM algorithm. This decomposition yields four simpler subproblems, several of which can be handled in parallel across antennas or subcarriers, and each can be solved either in closed-form or through a simple bisection search.
	
	\item The RF-precoder update (partially-connected) and the users' analog-combiner updates (fully- or partially-connected) remain nonconvex, even though the objective is convex with respect to each of these blocks. The nonconvexity is driven by the unit-modulus constraint imposed by every PS. To optimize the PS values at both the transmitter and the users, we deploy a coordinate descent strategy that updates a single PS while keeping the others fixed. For each PS, we derive a closed-form update rule, enabling scalable and low-cost implementations for large arrays.
	
	\item We further address the practical regime in which PSs are impaired and exact phase settings are difficult to implement. Under random phase errors, we propose robust PS optimization rules at the transmitter and the users, and derive closed-form updates that directly minimize the expected MSE in the presence of these imperfections. Finally, we also obtain closed-form solutions for the users' digital combiners.

\item In the overall BCD procedure, all variable blocks are updated in a Gauss--Seidel fashion until convergence. We provide theoretical support for the convergence behavior of the proposed updates for each subproblem, as well as for the convergence of the overall BCD routine. We also analyze the computational complexity of the proposed methods. Extensive simulations then demonstrate the practical impact of our approach, showing effective suppression of OOB emissions even under aggressive spectral masks and clipping, along with clear achievable-rate improvements over well-known benchmark schemes.
\end{itemize}

The rest of the paper is organized as follows. Section~\ref{sec:system} introduces the system model. Section~\ref{sec:probform} presents the problem formulation. Section~\ref{sec:alg} provides the proposed solution, including reformulations to handle the inherent complexities, the algorithmic decomposition and the corresponding optimization methods. Simulation results are reported and discussed in Section~\ref{sec:sim}. Finally, Section~\ref{sec:con} concludes the paper.
	
\section{System Model}
In this section, we define and study different components of the system model.
\label{sec:system}

\subsubsection{Hybrid RF-Digital Systems}

We consider the downlink of an MU-MIMO-OFDM system in which a BS with $N_t$ transmit antennas serves $K$ users over $S$ subcarriers. The subcarrier index set is defined as $\mathcal{S}=\{0,1,\ldots,S-1\}$. Each user is equipped with $N_r$ receive antennas and is scheduled on all subcarriers. A partially-connected hybrid architecture is assumed at the transmitter, where each antenna is connected to one RF chain and $N_{\text{RF}}$ denotes the number of RF chains. The $N_t$ antennas are partitioned into $N_{\text{RF}}$ disjoint subarrays of identical size, each containing $N_t/N_{\text{RF}}\in \mathbb{Z}$ antennas.

The BS applies an RF precoder (by PSs) in conjunction with per-subcarrier digital precoding. The RF precoding matrix is denoted by $\mathbf{V}_{\text{RF}}\in\mathbb{C}^{N_t \times N_{\text{RF}} }$. Owing to the partially-connected constraint, every row of $\mathbf{V}_{\text{RF}}$ contains exactly one non-zero complex entry with unit modulus. In particular, if the $a^\text{th}$ antenna is connected to the $m_a^\text{th}$ RF chain, then $\mathbf{V}_{\text{RF}}[a,m_a]=e^{\jmath \varphi_{a m_a}}$, where $\jmath=\sqrt{-1}$ and $\varphi_{a m_a}$ is the applied phase rotation. For subcarrier $s$, the digital precoder used to deliver $n_k$ streams to user $k$ is $\mathbf{V}_{k}^s\in\mathbb{C}^{N_{\text{RF}} \times n_k }$, and the corresponding transmitted symbol vector is $\boldsymbol{\omega}_{k}^s \in \mathbb{C}^{n_k \times 1}$. We model $\boldsymbol{\omega}_{k}^s\in \boldsymbol{\Omega}$ as a zero-mean normalized QAM symbol vector with $\mathbb E[\boldsymbol{\omega}_j^s(\boldsymbol{\omega}_j^s)^H]=\mathbf I_{n_j}$ and $\mathbb E[\boldsymbol{\omega}_j^s(\boldsymbol{\omega}_\ell^s)^H]=\mathbf 0$ for $j\neq \ell$. Once realized, the tuple $\boldsymbol{\omega}\triangleq\{\boldsymbol{\omega}_k^s\}_{k,s}$ is assumed known at the transmitter, which allows symbol-dependent digital precoding. We consider a batch of $B$ independent symbol realizations $\{\boldsymbol{\omega}^{(b)}\}_{b=1}^B$. Define the realized transmit vector on subcarrier \(s\) as $
\mathbf t^s \triangleq \sum_{j=1}^K \mathbf V_j^s\boldsymbol\omega_j^s \in \mathbb C^{N_{\text{RF}} }.$ We write $\mathbf V_k^s=\mathbf V_k^s(\boldsymbol{\omega})$, $\mathbf{V}_{\text{RF}}=\mathbf{V}_{\text{RF}}(\boldsymbol{\omega})$,  and $\mathbf t^s=\mathbf t^s(\boldsymbol{\omega})$, and the argument is suppressed except where an expectation over $\boldsymbol{\omega}$ is taken.

Following the linear model, the received signal $\mathbf{y}_{k}^s \in \mathbb{C}^{N_r \times 1}$ at user $k$ on subcarrier $s$ is expressed as
\begin{equation}
	\mathbf{y}_k^s = \mathbf{H}_k^s\mathbf{V}_{\text{RF}} \sum_{j=1}^K \mathbf{V}_j^s \boldsymbol{\omega}_j^s + \mathbf{n}_k^s = \mathbf{H}_k^s\mathbf{V}_{\text{RF}} \mathbf t^s  + \mathbf{n}_k^s,\label{eq:recsig}
\end{equation}
where $\mathbf{H}_k^s \in \mathbb{C}^{N_r \times N_t}$ denotes the channel matrix from the BS to user $k$ on subcarrier $s$, and $\mathbf{n}_k^s \in \mathbb{C}^{N_r \times 1}$ is additive white complex Gaussian noise distributed as $\mathcal{CN}(\mathbf{0}, \sigma_{\text{noise},k}^{s\,2}\mathbf{I}_{N_r})$.

At the receiver, user $k$ employs an analog combiner followed by a digital combiner. The decoded signal vector on subcarrier $s$ at the $k^\text{th}$ user is given by $\hat{\boldsymbol{\omega}}_k^s = \mathbf{U}_k^{s^H}\mathbf{U}_{\text{RF},k}^{H} \mathbf{y}_k^s,$ where $\mathbf{U}_{k}^s\in\mathbb{C}^{N_{\text{RF},k} \times n_k }$ and $\mathbf{U}_{\text{RF},k}\in\mathbb{C}^{N_r \times N_{\text{RF},k} }$ represent the digital and analog combining matrices, respectively. Since the receiver does not know the realized tuple $\boldsymbol{\omega}$, the hybrid combiners $\mathbf U_{\text{RF},k}$ and $\mathbf U_k^s$ are taken to be symbol-agnostic. Let $\mathcal E_k\subseteq \{1,\ldots,N_r\}\times\{1,\ldots,N_{\text{RF},k}\}$ denote the set of feasible nonzero entries of $\mathbf U_{\text{RF},k}$. Then, $\mathbf U_{\text{RF},k}[a,m]=e^{\jmath\varphi_{k,am}}$ for $(a,m)\in\mathcal E_k$, and $\mathbf U_{\text{RF},k}[a,m]=0$ for $(a,m)\notin\mathcal E_k$. We consider an arbitrary (partially or fully-connected) architecture for the combiner matrix of the receiver.

\subsubsection{Power Constraint for Hybrid Precoding}
The BS transmit budget on each subcarrier is limited according to
\begin{align}
	& \quad \norm{\mathbf{V}_{\text{RF}} \mathbf t^s }_2^2 \leq P^s.\label{eq:rfpowerbudget}
\end{align}
where $P^s$ denotes the total power available at the BS on subcarrier $s$.

In this paper, in addition to the conventional transmit-power limitation, we also incorporate clipping and spectral mask constraints that arise in practical MIMO-OFDM deployments.

\subsubsection{Clipping Constraint}

With an oversampling factor $\ell$, the CP-inclusive OFDM symbol emitted by the $a^\text{th}$ antenna is described in discrete time over
$n\in\{-\ell N_{\mathrm{CP}},-\ell N_{\mathrm{CP}}+1,\ldots,\ell S-1\}$. Over this interval, the transmitted samples follow \cite{van2009sculpting}
\begin{equation}
	\tilde{x}_{\text{CP}}^a[n] = \mathbf{p}^T[n]\mathbf{g}^a
	= \sum_{s\in \mathcal{S}} p^{s}[n] \mathbf{g}^a[s],\label{eq:ofdm_sym}
\end{equation}
where $\mathbf{p}[n] = [p^{0}[n], \ldots, p^{S-1}[n]]^T$ collects the subcarrier-modulated pulses (sampled on the $\ell S$-point grid) and $\mathbf{g}^a\in\mathbb{C}^{S}$ stacks the corresponding precoded frequency-domain symbols associated with the $a^\text{th}$ antenna. For the hybrid-precoding architecture, the precoded symbol vector is
\begin{equation}
	\mathbf g^a
	=
	\big[\mathbf V_{\mathrm{RF}}[a,:]\mathbf t^0,\;\ldots,\;\mathbf V_{\mathrm{RF}}[a,:]\mathbf t^{S-1}\big]^T
	= \mathbf V_{\mathrm{RF}}[a,m_a]\mathbf w^a,
\end{equation}
where $m_a$ is the (unique) RF-chain index connected to antenna $a$ (i.e., $\mathbf{V}_\text{RF}[a,m_a]\neq 0$), and
\begin{align}
	\mathbf{w}^a[s]=\mathbf t^s[m_a], \quad \forall a,\ \forall s\in\mathcal S. \label{eq:wa_def}
\end{align}
Hereafter, $\mathbf V_k^s$, $\mathbf w^a$, and $\mathbf g^a$ all implicitly depend on the realized tuple $\boldsymbol{\omega}$, and we suppress that dependence to keep the notation easy to read.

The discrete-time OFDM pulse with subcarrier modulation for the $s^\text{th}$ subcarrier is defined as \cite{van2009sculpting,9214877,9390405}
\begin{equation}
	p^s[n] = \frac{1}{\sqrt{\ell S}}\,e^{\jmath 2\pi \frac{s}{\ell S} n}\, I[n], \:\:\: s\in \mathcal{S},\label{eq:timedomain}
\end{equation}
where $I[n]$ is the indicator function given by $I[n] = 1$ for $-\ell N_{\mathrm{CP}} \leq n \leq \ell S-1$ and $I[n] = 0$ otherwise, with $N_{\mathrm{CP}}$ representing the CP length (in samples). The rectangular window assumption provides a worst-case characterization of OOB emissions since it yields the slowest spectral decay among common pulse shapes; thus, any suppression demonstrated here is expected to be at least as good when band-limited pulses are employed.

Fix $a\in\{1,\ldots,N_t\}$ and consider an $\ell S$-point IDFT grid. Define the (oversampled) IDFT matrix
$\mathbf F_{\ell S}^{H}\in\mathbb C^{\ell S\times S}$ with entries
\begin{equation}
	\mathbf F_{\ell S}^{H}[n,s]
	\triangleq \frac{1}{\sqrt{\ell S}}e^{\jmath 2\pi \frac{ns}{\ell S}},
	\qquad n\in\{0,1,\ldots,\ell S-1\},\ \ s\in\mathcal{S},
	\nonumber
\end{equation}
so that the useful (no-CP) time-domain block $\tilde{\mathbf x}^{a}\in\mathbb C^{\ell S}$ is
\begin{equation}
	\tilde{\mathbf x}^{a}=\mathbf F_{\ell S}^{H}\mathbf g^{a}
	= \mathbf{V}_\text{RF}[a,m_a]\mathbf F_{\ell S}^{H}\mathbf w^{a}.
	\nonumber
\end{equation}
 Let \(\chi\) denote the clipping level. It is desired that $|\tilde{x}^{a}[n]| \leq \chi$ for \(0 \leq n \leq \ell S-1\) or equivalently $ \| \tilde{\mathbf{x}}^a \|_{\infty} \leq \chi$, so as to avoid signal clipping. Note that since $| \mathbf{V}_\text{RF}[a,m_a]|=1$, then $ \| \tilde{\mathbf{x}}^a \|_{\infty} = \| \mathbf F_{\ell S}^{H}\mathbf w^{a} \|_{\infty}$. Since the CP is formed by copying the last $\ell N_{\mathrm{CP}}$ samples of the useful OFDM block, enforcing $|\tilde{x}^a[n]|\le \chi$ for $0\le n\le \ell S-1$ automatically guarantees the same bound on the CP samples.

\subsubsection{Spectral Mask Constraint}

In MIMO-OFDM, the sidelobes of the subcarrier waveforms spread energy outside the intended band, potentially causing adjacent-channel interference and violating regulatory emission masks. Spectral mask constraints address this by enforcing suppression at selected out-of-band frequencies, particularly near the band edges.

Fix an oversampling factor $\ell$ and consider the CP-inclusive transmit pulse $p^{s}[n]$ in \eqref{eq:timedomain} over
$n\in\{-\ell N_{\mathrm{CP}},\ldots,\ell S-1\}$. Let $F_{s,\ell}$ denote the sampling rate of the oversampled discrete-time
waveform. For spectral sampling, introduce a set of (possibly non-integer) locations
$\{\gamma_i\}_{i=0}^{M-1}\subset\mathbb R$ expressed in DFT-bin units on the $\ell S$ grid; the location $\gamma_i$
corresponds to the tangible baseband frequency $f_i=\frac{\gamma_i}{\ell S}F_{s,\ell}$ (Hz). Integer $\gamma_i$ coincide with
$\ell S$-point DFT-bin centers, whereas non-integer $\gamma_i$ evaluate the DFT between bins.

Using the rectangular window in \eqref{eq:timedomain}, define the sampling matrix $\mathbf A\in\mathbb C^{M\times S}$ by
evaluating that formula at each $\gamma_i$, i.e.,
\begin{align}
	\mathbf A[i,s]
	=&
	\frac{1}{\sqrt{\ell S}}
	\exp\!\left(\jmath\pi\frac{\gamma_i-s}{\ell S}\,(\ell N_{\mathrm{CP}}-\ell S+1)\right)\nonumber\\
	&\times
	\frac{\sin\!\left(\pi\frac{\gamma_i-s}{\ell S}(\ell S+\ell N_{\mathrm{CP}})\right)}
	{\sin\!\left(\pi\frac{\gamma_i-s}{\ell S}\right)},
	\qquad \gamma_i-s\notin \ell S\,\mathbb Z,
	\label{eq:A_closed}
\end{align}
and $\mathbf A[i,s]=L/\sqrt{\ell S}$ when $\gamma_i-s\in \ell S\,\mathbb Z$, with $L\triangleq \ell S+\ell N_{\mathrm{CP}}$ \cite{9214877,van2009sculpting}.

For any given antenna $a$, \eqref{eq:ofdm_sym}--\eqref{eq:timedomain} imply that the CP-inclusive spectrum evaluated at the sampled
locations satisfies
\begin{align}
	X^a(\gamma_i)
	\;&\triangleq\;
	\sum_{n=-\ell N_{\mathrm{CP}}}^{\ell S-1} \tilde{x}_{\text{CP}}^a[n]\;e^{-\jmath 2\pi \frac{\gamma_i}{\ell S}n}
	\;=\;
	\sum_{s\in\mathcal S}\mathbf g^a[s]\;\mathbf A[i,s]
	\;\nonumber\\
	&=\;
	\mathbf A[i,:]\mathbf g^a,
	\qquad i=0,\ldots,M-1,
	\nonumber
\end{align}
and a (single-symbol) periodogram-type PSD sample at $\gamma_i$ is
\begin{equation}
	\widehat S_{\tilde{x}_{\text{CP}}^a \tilde{x}_{\text{CP}}^a}(\gamma_i)
	\;\triangleq\;
	\frac{1}{L\,F_{s,\ell}}\bigl|X^a(\gamma_i)\bigr|^2
	\;=\;
	\frac{1}{L\,F_{s,\ell}}\bigl|\mathbf A[i,:]\mathbf g^a\bigr|^2.
	\label{eq:PSD_inst_A}
\end{equation}
In the following, we use the single-symbol spectrum sample (periodogram-type) at frequency location $\gamma$,
$\widehat S_{\tilde{x}_{\text{CP}}^a \tilde{x}_{\text{CP}}^a}(\gamma)= \frac{1}{L F_{s,\ell}}|X^a(\gamma)|^2$,
as a measure of spectral leakage. The emission-mask constraints are enforced directly on these samples at a finite set of mask frequencies.
Under $\mathbf g^a=\mathbf V_{\text{RF}}[a,m_a]\mathbf w^a$, one has
$X^a(\gamma_i)=\mathbf V_{\text{RF}}[a,m_a]\;\mathbf A[i,:]\mathbf w^a$ and therefore
$|X^a(\gamma_i)|^2 = |\mathbf V_{\text{RF}}[a,m_a]|^2\,|\mathbf A[i,:]\mathbf w^a|^2$; in particular, if
$|\mathbf V_{\text{RF}}[a,m_a]|=1$, then the PSD samples in \eqref{eq:PSD_inst_A} are unaffected by the analog precoder's magnitude. Hence, without loss of generality, we define $\mathbf x^a \triangleq \mathbf F_{\ell S}^{H}\mathbf w^a$ and impose clipping constraint using $\mathbf x^a$, since $\| \tilde{\mathbf x}^a \|_{\infty}=\| \mathbf{x}^a \|_{\infty}$.

The spectral mask is imposed by constraining the emitted spectrum at a finite set of mask frequencies
$\{f_1,\ldots,f_G\}$. These mask frequencies are different from the generic sampling locations
$\{\gamma_i\}_{i=0}^{M-1}$ used for PSD evaluation/visualization: the latter can be chosen dense over a wide span to inspect the
spectrum, whereas $\{f_j\}$ is a design set used only to enforce compliance. Following \cite{van2009sculpting}, the mask points
are placed tightly near band edges and/or around interference regions; closely-spaced samples in these critical areas provide an
effective discrete surrogate of continuous mask constraints.

Map each physical mask frequency $f_j$ to its DFT-bin location on the $\ell S$ grid via
$\gamma_j \triangleq \frac{\ell S}{F_{s,\ell}}f_j$. Let $\mathbf A_n\in\mathbb C^{G\times S}$ denote the matrix
obtained by evaluating \eqref{eq:A_closed} at $\{\gamma_j\}_{j=1}^{G}$ (equivalently, selecting the corresponding rows of
$\mathbf A$ when $\{\gamma_j\}\subset\{\gamma_i\}$). Then, under hybrid precoding, $
X^a(\gamma_j)=\mathbf V_{\text{RF}}[a,m_a]\;\mathbf A_n[j,:]\mathbf w^a,
\: j=1,\ldots,G.$
The mask constraints at antenna $a$ are imposed as
\begin{equation}
	\frac{\bigl|\mathbf A_n\mathbf w^a\bigr|^2}{L\,F_{s,\ell}} \preceq \frac{\mathbf r}{L\,F_{s,\ell}},
	\qquad \mathbf r=[r_1,\ldots,r_G]^T,
	\label{eq:mask_constraint}
\end{equation}
where $\mathbf r[j] \triangleq L\,F_{s,\ell}\,S_{\text{max}}(f_j)$, $S_{\text{max}}(f_j)$ is the maximum allowable PSD and $\mathbf r[j]$ denotes one sample at $f_j$, i.e., $|X^a(\gamma_j)|^2 \le \mathbf r[j]$. Equivalently, \eqref{eq:mask_constraint} requires designing the digital precoders $\{\mathbf V_k^s\}$ such that the resulting per-antenna vector $\mathbf w^a$ yields OFDM symbols whose spectral samples satisfy the prescribed limits at $\{f_1,\ldots,f_G\}$.

\section{Problem Formulation}
\label{sec:probform}
The objective of this paper is to minimize the sum-MSE of the decoded symbol vectors aggregated over all users and
subcarriers, under the assumption that every user is served on every $s\in\mathcal S$. In particular, the estimation error of
user $k$ on subcarrier $s$ is defined as
\begin{equation}
	\mathbf e_k^s \triangleq \hat{\boldsymbol{\omega}}_k^s-\boldsymbol{\omega}_k^s
	= \mathbf{U}_k^{s^H}\mathbf{U}_{\text{RF},k}^{H}\mathbf{y}_k^s-\boldsymbol{\omega}_k^s,
\end{equation}
where $\hat{\boldsymbol{\omega}}_k^s=\mathbf{U}_k^{s^H}\mathbf{U}_{\text{RF},k}^{H} \mathbf{y}_k^s$ and $\mathbf y_k^s$ is given in
\eqref{eq:recsig}. We then consider the MMSE criterion $ \sum_{s\in\mathcal S}\sum_{k=1}^K \mathbb{E}_{\boldsymbol{\omega},\mathbf n}\!\left[\left\|\mathbf e_k^s\right\|_2^2\right]$, where the expectation is taken with respect to the additive noise $\{\mathbf n_k^s\}$ and the empirical distribution of the symbol tuple $\boldsymbol{\omega}$ induced by the fixed batch $\{\boldsymbol{\omega}^{(b)}\}_{b=1}^B$. At the same time, the receiver-side hybrid combiners remain symbol-agnostic, whereas the digital transmit-side variables are designed pointwise once $\boldsymbol{\omega}$ is realized.

We formulate the joint hybrid precoding problem to minimize the sum-MSE, subject to transmit power budget, clipping, and
spectral mask constraints as follows:
\begin{subequations}\label{opt:main}
	\begin{alignat}{2}
		\min_{\substack{\{\mathbf V_k^s\},\,\mathbf V_{\text{RF}},\,\{\mathbf U_k^s\},\\[2pt] \{\mathbf U_{\text{RF},k}\},\,\{\mathbf x^a\},\,\{\mathbf w^a\}}}
		&\quad J \triangleq \sum_{s\in\mathcal S}\sum_{k=1}^K \mathbb{E}_{\boldsymbol{\omega},\mathbf n}\!\left[\left\|\mathbf e_k^s\right\|_2^2\right] \nonumber\\
		\text{s.t.}\quad
		& \eqref{eq:rfpowerbudget},\ \eqref{eq:wa_def},\ \eqref{eq:mask_constraint}, \nonumber\\
		& \mathbf{x}^a = \mathbf F_{\ell S}^{H} \mathbf{w}^a, &\quad& \forall a, \label{eq:xa}\\
		& \| \mathbf{x}^a \|_{\infty} \leq \chi, &\quad& \forall a. \label{eq:clipping_constraint_digital}
	\end{alignat}
\end{subequations}
The constraints are instantaneous and are therefore enforced pointwise in $\boldsymbol{\omega}$, while the objective is averaged over $\boldsymbol{\omega}$ and $\mathbf n$.

\section{The Proposed BCD-Based Algorithm}
\label{sec:alg}
Note that $J$ can be written as $J=\sum_{s,k}\mathbb E_{\boldsymbol{\omega}}[\|\mathbf U_k^{s^H}\mathbf U_{\mathrm{RF},k}^H\mathbf H_k^s\mathbf V_{\mathrm{RF}}\mathbf t^s-\boldsymbol{\omega}_k^s\|_2^2]+\sum_{s,k}\sigma_{\mathrm{noise},k}^{s\,2}\tr(\mathbf{U}_k^{s^H}\mathbf{U}_{\text{RF},k}^H\mathbf{U}_{\text{RF},k}\mathbf{U}_k^s)$. At each BCD iteration, the $B$ transmit-side subproblems are independent across realizations and solved in parallel via the ADMM of Section~\ref{subsec:admm_blocks_revised_hyb}; the resulting $\{\mathbf t^s(\boldsymbol{\omega}^{(b)})\}_{b=1}^B$ are used to form the sample-average covariances for the combiner update.  Then, for a fixed realized tuple $\boldsymbol{\omega}$,
\begin{align}
	&	\mathbf E_k^s
	\triangleq \mathbb{E}_{\mathbf n}\!\left[(\hat{\boldsymbol{\omega}}_k^s-\boldsymbol{\omega}_k^s)(\hat{\boldsymbol{\omega}}_k^s-\boldsymbol{\omega}_k^s)^H\right]\nonumber\\
	&=
	\mathbf{U}_k^{s^H}\mathbf{U}_{\text{RF},k}^{H}\mathbf{H}_k^s\mathbf{V}_{\text{RF}}\mathbf t^s(\mathbf t^s)^H
	\mathbf V_{\text{RF}}^{H}(\mathbf H_k^s)^{H}\mathbf U_{\text{RF},k}\mathbf U_k^s\nonumber\\
	&-\mathbf{U}_k^{s^H}\mathbf{U}_{\text{RF},k}^{H}\mathbf{H}_k^s\mathbf{V}_{\text{RF}}\mathbf t^s(\boldsymbol{\omega}_k^s)^H
	-\boldsymbol{\omega}_k^s(\mathbf t^s)^H\mathbf V_{\text{RF}}^{H}(\mathbf H_k^s)^{H}\mathbf U_{\text{RF},k}\mathbf U_k^s\nonumber\\
	&+\boldsymbol{\omega}_k^s(\boldsymbol{\omega}_k^s)^H
	+\sigma_{\text{noise},k}^{s\,2}\mathbf{U}_k^{s^H}\mathbf{U}_{\text{RF},k}^H\mathbf{U}_{\text{RF},k}\mathbf{U}_k^s.
	\label{eq:E_det_t_expand}
\end{align}
With the above reformulation, \eqref{opt:main} is equivalently written as
\begin{align}
	\min_{\substack{\{\mathbf U_k^s\},\,\{\mathbf t^s\},\,\mathbf{V}_{\text{RF}},\\[2pt] \{\mathbf U_{\text{RF},k}\},\,\{\mathbf{x}^a\},\,\{\mathbf{w}^a\}}}
	&\quad \sum_{s\in\mathcal S}\sum_{k=1}^K \mathbb E_{\boldsymbol{\omega}}\!\left[\tr(\mathbf E_k^s)\right] \label{opt:main-revised}\\
	\text{s.t.}\quad
	&\ \eqref{eq:rfpowerbudget},\ \eqref{eq:wa_def},\ \eqref{eq:mask_constraint},\ \eqref{eq:xa},\ \eqref{eq:clipping_constraint_digital},\nonumber
\end{align}
where $\text{tr}(\cdot)$ is the trace operator. Here, $\mathbf E_k^s$ is the noise-averaged MSE matrix for the current realization $\boldsymbol{\omega}$, so the remaining expectation is only over $\boldsymbol{\omega}$.

\subsection{A Four-Block ADMM Reformulation without Explicit $\{\mathbf V_k^s\}$}
\label{subsec:admm_blocks_revised_hyb}

For fixed combiners $\{\mathbf U_k^s\}$ and $\{\mathbf U_{\mathrm{RF},k}\}$, the transmit-side subproblem is solved independently for each realization $\boldsymbol{\omega}^{(b)}$, $b=1,\ldots,B$, yielding the corresponding $\mathbf t^{s(b)}$, $\mathbf x^{a(b)}$, $\mathbf w^{a(b)}$, and $\mathbf V_{\mathrm{RF}}^{(b)}$; the $B$ instances may be solved in parallel. In what follows, we drop the superscript $(b)$ from $\mathbf t^{s(b)}$, $\mathbf x^{a(b)}$, $\mathbf w^{a(b)}$, and $\mathbf V_{\mathrm{RF}}^{(b)}$ to simplify the notation, and present the solution for a single realization $\boldsymbol{\omega}$.

For the partially connected architecture, $\mathbf w^a[s]=\mathbf t^s[m_a]$ for all $a$ and $s$. Moreover, because the nonzero entries of
$\mathbf V_{\mathrm{RF}}$ have unit modulus and each RF chain is connected to $N_t/N_{\mathrm{RF}}$ antennas, we have \cite{7913599} $
\mathbf V_{\mathrm{RF}}^H\mathbf V_{\mathrm{RF}}=\frac{N_t}{N_{\mathrm{RF}}}\mathbf I_{N_{\mathrm{RF}}}.$
Hence, the inner problem may be expressed entirely through $\{\mathbf t^s\}$, $\{\mathbf w^a\}$, and $\{\mathbf x^a\}$. To handle the spectral-mask
constraint, define $\mathbf q^a \triangleq \mathbf A_n\mathbf w^a \in \mathbb C^{G}$ for $a=1,\ldots,N_t$. Thus, the inner problem is written in
terms of $\{\mathbf t^s\}$, $\{\mathbf w^a\}$, $\{\mathbf x^a\}$, and $\{\mathbf q^a\}$. Once $\{\mathbf t^s\}$ is available, a compatible family
$\{\mathbf V_k^s\}$ can be reconstructed afterward. We also introduce two quadratic regularization
terms with parameters $\eta_w>0$ and $\eta_t>0$.

For fixed $\mathbf V_{\mathrm{RF}}$, $\{\mathbf U_{\mathrm{RF},k}\}$, and $\{\mathbf U_k^s\}$, let
\[
\mathbf B_k^s \triangleq \mathbf U_k^{s^H}\mathbf U_{\mathrm{RF},k}^H\mathbf H_k^s\mathbf V_{\mathrm{RF}}
\in \mathbb C^{n_k\times N_{\mathrm{RF}}},\qquad \forall k,s.
\]
The stacked variables are
$\mathbf T \triangleq [\mathbf t^0 \cdots \mathbf t^{S-1}]\in\mathbb C^{N_{\mathrm{RF}}\times S}$,
$\mathbf W \triangleq [\mathbf w^1 \cdots \mathbf w^{N_t}]^T \in \mathbb C^{N_t\times S}$,
$\mathbf X \triangleq [\mathbf x^1 \cdots \mathbf x^{N_t}]^T \in \mathbb C^{N_t\times \ell S}$, and
$\mathbf Q \triangleq [\mathbf q^1 \cdots \mathbf q^{N_t}]^T \in \mathbb C^{N_t\times G}$. By construction, $
\mathbf W[a,:]=\mathbf T[m_a,:], \forall a=1,\ldots,N_t.$
Since $\mathbf W$ is arranged row-wise, we introduce the operators $
\mathfrak F(\mathbf W)\triangleq \mathbf W(\mathbf F_{\ell S}^{H})^{T}\in\mathbb C^{N_t\times \ell S}$, $
\mathfrak A(\mathbf W)\triangleq \mathbf W\mathbf A_n^{T}\in\mathbb C^{N_t\times G}.$

Next, define the closed convex sets
\begin{align}
	\mathcal M &\triangleq \Big\{\mathbf Q\in\mathbb C^{N_t\times G}: |\mathbf Q[a,j]|^2\le r_j,\ \forall a,j\Big\},\nonumber\\
	\mathcal C &\triangleq \Big\{\mathbf X\in\mathbb C^{N_t\times \ell S}: |\mathbf X[a,n]|\le \chi,\ \forall a,n\Big\},\nonumber\\
	\mathcal P &\triangleq \Big\{\mathbf T\in\mathbb C^{N_{\mathrm{RF}}\times S}: \frac{N_t}{N_{\mathrm{RF}}}\|\mathbf T[:,s]\|_2^2\le P^s,\ \forall s\in\mathcal S\Big\}.\nonumber
\end{align}
For any closed convex set $\mathcal Z$, let $\delta_{\mathcal Z}(\mathbf Z)$ denote its indicator function, equal to $0$ if $\mathbf Z\in\mathcal Z$ and $+\infty$ otherwise.

With $\mathbf V_{\mathrm{RF}}$, $\{\mathbf U_{\mathrm{RF},k}\}$, and $\{\mathbf U_k^s\}$ fixed, the regularized inner problem becomes
\begin{align}
	\min_{\mathbf Q,\mathbf X,\mathbf W,\mathbf T}\quad
	&\theta_1(\mathbf Q)+\theta_2(\mathbf X)+\theta_3(\mathbf W)+\theta_4(\mathbf T)
	\label{opt:inner4block_hyb}\\
	\text{s.t.}\quad
	&\mathbf Q-\mathfrak A(\mathbf W)=\mathbf 0,\nonumber\\
	&\mathbf X-\mathfrak F(\mathbf W)=\mathbf 0,\nonumber\\
	&\mathbf W[a,:]-\mathbf T[m_a,:]=\mathbf 0,\qquad \forall a,\nonumber
\end{align}
where $\theta_1(\mathbf Q)\triangleq \delta_{\mathcal M}(\mathbf Q)$, $\theta_2(\mathbf X)\triangleq \delta_{\mathcal C}(\mathbf X)$, $\theta_3(\mathbf W)\triangleq \frac{\eta_w}{2}\|\mathbf W\|_F^2$, and $\theta_4(\mathbf T)\triangleq \sum_{s\in\mathcal S}\sum_{k=1}^K \|\mathbf B_k^s\mathbf T[:,s]-\boldsymbol\omega_k^s\|_2^2 +\delta_{\mathcal P}(\mathbf T)+\frac{\eta_t}{2}\|\mathbf T\|_F^2$.

Let $\langle \mathbf A,\mathbf B\rangle \triangleq \Re\{\tr(\mathbf A^H\mathbf B)\}$. Introducing the dual matrices
$\boldsymbol\Lambda_q\in\mathbb C^{N_t\times G}$, $\boldsymbol\Lambda_x\in\mathbb C^{N_t\times \ell S}$, and
$\boldsymbol\Lambda_w\in\mathbb C^{N_t\times S}$, the augmented Lagrangian is
\begin{align}
	&\mathcal L(\mathbf Q,\mathbf X,\mathbf W,\mathbf T,\boldsymbol\Lambda_q,\boldsymbol\Lambda_x,\boldsymbol\Lambda_w)=
	\theta_1(\mathbf Q)+\theta_2(\mathbf X)+\theta_3(\mathbf W)\nonumber\\
	& +\theta_4(\mathbf T)+\langle \boldsymbol\Lambda_q,\mathbf Q-\mathfrak A(\mathbf W)\rangle
	+\frac{\rho}{2}\|\mathbf Q-\mathfrak A(\mathbf W)\|_F^2\nonumber\\
	&+\langle \boldsymbol\Lambda_x,\mathbf X-\mathfrak F(\mathbf W)\rangle +\frac{\rho}{2}\|\mathbf X-\mathfrak F(\mathbf W)\|_F^2\nonumber\\
	&
	+\sum_{a=1}^{N_t}\left\langle (\boldsymbol\Lambda_w[a,:])^T,\ (\mathbf W[a,:])^T-(\mathbf T[m_a,:])^T\right\rangle\nonumber\\
	&+\frac{\rho}{2}\sum_{a=1}^{N_t}\|(\mathbf W[a,:])^T-(\mathbf T[m_a,:])^T\|_2^2.\nonumber
\end{align}
Accordingly, at iteration $\tau+1$, the cyclic ADMM steps are
\begin{equation}
	\label{eq:admm_4block_updates_hyb}
	\begin{split}
		&\mathbf Q^{\tau+1}
		=
		\arg\min_{\mathbf Q}\ \mathcal L\!\left(
		\mathbf Q,\mathbf X^\tau,\mathbf W^\tau,\mathbf T^\tau,
		\boldsymbol\Lambda_q^\tau,\boldsymbol\Lambda_x^\tau,\boldsymbol\Lambda_w^\tau
		\right),\\
		&\mathbf X^{\tau+1}
		=
		\arg\min_{\mathbf X}\ \mathcal L\!\left(
		\mathbf Q^{\tau+1},\mathbf X,\mathbf W^\tau,\mathbf T^\tau,
		\boldsymbol\Lambda_q^\tau,\boldsymbol\Lambda_x^\tau,\boldsymbol\Lambda_w^\tau
		\right),\\
		&\mathbf W^{\tau+1}
		=
		\arg\min_{\mathbf W}\ \mathcal L\!\left(
		\mathbf Q^{\tau+1},\mathbf X^{\tau+1},\mathbf W,\mathbf T^\tau,
		\boldsymbol\Lambda_q^\tau,\boldsymbol\Lambda_x^\tau,\boldsymbol\Lambda_w^\tau
		\right),\\
		&\mathbf T^{\tau+1}
		=
		\arg\min_{\mathbf T}\ \mathcal L\!\left(
		\mathbf Q^{\tau+1},\mathbf X^{\tau+1},\mathbf W^{\tau+1},\mathbf T,
		\boldsymbol\Lambda_q^\tau,\boldsymbol\Lambda_x^\tau,\boldsymbol\Lambda_w^\tau
		\right),\\
		&\boldsymbol\Lambda_q^{\tau+1}
		=
		\boldsymbol\Lambda_q^\tau+\rho\big(\mathbf Q^{\tau+1}-\mathfrak A(\mathbf W^{\tau+1})\big),\\
		&\boldsymbol\Lambda_x^{\tau+1}
		=
		\boldsymbol\Lambda_x^\tau+\rho\big(\mathbf X^{\tau+1}-\mathfrak F(\mathbf W^{\tau+1})\big),\\
		&(\boldsymbol\Lambda_w^{\tau+1}[a,:])^T
		=
		(\boldsymbol\Lambda_w^\tau[a,:])^T+\rho\big((\mathbf W^{\tau+1}[a,:])^T-(\mathbf T^{\tau+1}[m_a,:])^T\big).
	\end{split}
	\raisetag{-8pt}
\end{equation}

\subsubsection{Update of $\mathbf Q$}
\label{subsec:q_solution_hyb}

The $\mathbf Q$-subproblem reduces to the Euclidean projection onto $\mathcal M$:
\[
\mathbf Q^{\tau+1}
=
\operatorname{Proj}_{\mathcal M}\!\left(\mathfrak A(\mathbf W^\tau)-\frac{1}{\rho}\boldsymbol\Lambda_q^\tau\right).
\]
Equivalently, for every antenna $a$ and mask sample $j$,
\begin{align}
	&\mathbf q^{a,\tau+1}[j]
	=
	\min\!\Bigg(
	1,\frac{\sqrt{r_j}}{\left|\mathbf A_n[j,:]\mathbf w^{a,\tau}-\boldsymbol\Lambda_q^\tau[a,j]/\rho\right|}
	\Bigg)\label{eq:q_closed_hyb}\\
	&\qquad\times
	\left(\mathbf A_n[j,:]\mathbf w^{a,\tau}-\boldsymbol\Lambda_q^\tau[a,j]/\rho\right),
	\qquad j=1,\ldots,G,
	\nonumber
\end{align}
with the convention that $\mathbf q^{a,\tau+1}[j]=0$ whenever the term inside parentheses is zero.

\subsubsection{Update of $\mathbf X$}
\label{subsec:x_solution_hyb}

The $\mathbf X$-subproblem is likewise a Euclidean projection, now onto $\mathcal C$:
\[
\mathbf X^{\tau+1}
=
\operatorname{Proj}_{\mathcal C}\!\left(\mathfrak F(\mathbf W^\tau)-\frac{1}{\rho}\boldsymbol\Lambda_x^\tau\right).
\]
In elementwise form, for each antenna $a$ and sample index $n=0,\ldots,\ell S-1$,
\begin{align}
	&\mathbf x^{a,\tau+1}[n]
	=
	\min\!\Bigg(
	1,\frac{\chi}{\left|\mathbf F_{\ell S}^{H}[n,:]\mathbf w^{a,\tau}-\boldsymbol\Lambda_x^\tau[a,n]/\rho\right|}
	\Bigg)\nonumber\\
	&\qquad\times
	\left(\mathbf F_{\ell S}^{H}[n,:]\mathbf w^{a,\tau}-\boldsymbol\Lambda_x^\tau[a,n]/\rho\right).
	\label{eq:x_closed_hyb}
\end{align}

\subsubsection{Update of $\mathbf W$}
\label{subsec:w_solution_hyb}

The minimization with respect to $\mathbf W$ is unconstrained, strongly convex, and separable across antennas. For each
antenna, the first-order optimality condition gives
\begin{align}
	&\Big((\eta_w+\rho)\mathbf I_S+\rho\mathbf A_n^H\mathbf A_n+\rho\mathbf F_{\ell S}\mathbf F_{\ell S}^{H}\Big)\mathbf w^{a,\tau+1}
	\nonumber\\
	&=
	\mathbf A_n^H\!\big(\rho\mathbf q^{a,\tau+1}+(\boldsymbol\Lambda_q^\tau[a,:])^T\big)
	+\mathbf F_{\ell S}\!\big(\rho\mathbf x^{a,\tau+1}+(\boldsymbol\Lambda_x^\tau[a,:])^T\big)\nonumber\\
	&\quad+\rho(\mathbf T^\tau[m_a,:])^T-(\boldsymbol\Lambda_w^\tau[a,:])^T.
	\nonumber
\end{align}
Using $\mathbf F_{\ell S}\mathbf F_{\ell S}^{H}=\mathbf I_S$, this simplifies to
\begin{align}
	&\mathbf M_w\,\mathbf w^{a,\tau+1}
	=
	\mathbf A_n^H\!\big(\rho\mathbf q^{a,\tau+1}+(\boldsymbol\Lambda_q^\tau[a,:])^T\big)
	\label{eq:w_normal_per_antenna_hyb}\\
	&+\mathbf F_{\ell S}\!\big(\rho\mathbf x^{a,\tau+1}+(\boldsymbol\Lambda_x^\tau[a,:])^T\big)+\rho(\mathbf T^\tau[m_a,:])^T-(\boldsymbol\Lambda_w^\tau[a,:])^T,
	\nonumber
\end{align}
where $\mathbf M_w \triangleq (\eta_w+2\rho)\mathbf I_S+\rho\mathbf A_n^H\mathbf A_n\in\mathbb C^{S\times S}$.

Since $\mathbf M_w$ is the same for all antennas, its factorization can be computed once and reused for all $\{\mathbf w^{a,\tau+1}\}_{a=1}^{N_t}$. Moreover, when $G\ll S$, letting $\upsilon\triangleq \eta_w+2\rho$, one may use the matrix inversion lemma:
\[
\mathbf M_w^{-1}
=
\upsilon^{-1}\mathbf I_S-\upsilon^{-2}\mathbf A_n^H\big(\rho^{-1}\mathbf I_G+\upsilon^{-1}\mathbf A_n\mathbf A_n^H\big)^{-1}\mathbf A_n.
\]
Hence, $\mathbf w^{a,\tau+1}=\mathbf M_w^{-1}\mathbf b^{a,\tau}$ can be computed by inverting a $G\times G$ matrix instead of an $S\times S$ matrix, where $\mathbf b^{a,\tau}$ denotes the right-hand side of \eqref{eq:w_normal_per_antenna_hyb}. This is exact and is particularly attractive when $G$ is much smaller than $S$.

\subsubsection{Update of $\mathbf T$}
\label{subsec:t_solution_hyb}

The $\mathbf T$-subproblem is strongly convex and separates over $s\in\mathcal S$. For each $s\in\mathcal S$,
\begin{align}
	\min_{\mathbf t^s}\quad
	&\sum_{k=1}^{K}\|\mathbf B_k^s\mathbf t^s-\boldsymbol\omega_k^s\|_2^2
	+\frac{\eta_t}{2}\|\mathbf t^s\|_2^2
	\label{opt:t_sub_hyb}\\
	&+\sum_{a=1}^{N_t}\Re\!\left\{(\boldsymbol\Lambda_w^\tau[a,s])^*\big(\mathbf W^{\tau+1}[a,s]-\mathbf t^s[m_a]\big)\right\}\nonumber\\
	&+\frac{\rho}{2}\sum_{a=1}^{N_t}\big|\mathbf W^{\tau+1}[a,s]-\mathbf t^s[m_a]\big|^2
	\quad
	\text{s.t.}\quad
	\frac{N_t}{N_{\mathrm{RF}}}\|\mathbf t^s\|_2^2\le P^s .\nonumber
\end{align}
Ignoring the power constraint, the unique minimizer satisfies
\begin{align}
	&\Bigg(2\sum_{k=1}^{K}(\mathbf B_k^s)^H\mathbf B_k^s+\Big(\eta_t+\rho\frac{N_t}{N_{\mathrm{RF}}}\Big)\mathbf I_{N_{\mathrm{RF}}}\Bigg)\mathbf t^s
	\label{eq:t_normal_hyb}\\
	&=
	2\sum_{k=1}^{K}(\mathbf B_k^s)^H\boldsymbol\omega_k^s
	+
	\begin{bmatrix}
		\sum_{a:m_a=1}\big(\rho\mathbf W^{\tau+1}[a,s]+\boldsymbol\Lambda_w^\tau[a,s]\big)\\
		\vdots\\
		\sum_{a:m_a=N_{\mathrm{RF}}}\big(\rho\mathbf W^{\tau+1}[a,s]+\boldsymbol\Lambda_w^\tau[a,s]\big)
	\end{bmatrix}.
	\nonumber
\end{align}
If the solution to \eqref{eq:t_normal_hyb} satisfies $\frac{N_t}{N_{\mathrm{RF}}}\|\mathbf t^s\|_2^2\le P^s$, then it is also optimal for
\eqref{opt:t_sub_hyb}. Otherwise, the power constraint is active. Let $\mu^s\ge0$ be the corresponding Lagrange multiplier and define
\[
\mathbf M_t^s(\mu^s)\triangleq
2\sum_{k=1}^{K}(\mathbf B_k^s)^H\mathbf B_k^s
+\Big(\eta_t+(\rho+2\mu^s)\frac{N_t}{N_{\mathrm{RF}}}\Big)\mathbf I_{N_{\mathrm{RF}}}
\in\mathbb C^{N_{\mathrm{RF}}\times N_{\mathrm{RF}}}.
\]
Then,
\begin{align}
	\mathbf t^s(\mu^s)
	&=
	\big(\mathbf M_t^s(\mu^s)\big)^{-1}
	\Bigg(
	2\sum_{k=1}^{K}(\mathbf B_k^s)^H\boldsymbol\omega_k^s
	\nonumber\\
	&\qquad+
	\begin{bmatrix}
		\sum_{a:m_a=1}\big(\rho\mathbf W^{\tau+1}[a,s]+\boldsymbol\Lambda_w^\tau[a,s]\big)\\
		\vdots\\
		\sum_{a:m_a=N_{\mathrm{RF}}}\big(\rho\mathbf W^{\tau+1}[a,s]+\boldsymbol\Lambda_w^\tau[a,s]\big)
	\end{bmatrix}
	\Bigg),
	\nonumber
\end{align}
where $\mu^s>0$ is obtained via bisection search such that $
\frac{N_t}{N_{\mathrm{RF}}}\|\mathbf t^s(\mu^s)\|_2^2=P^s.$

The updates for $\mathbf Q$, $\mathbf X$, and $\mathbf W$ are separable across antennas, whereas the update for $\mathbf T$ separates across
subcarriers.

\begin{proposition}
	\label{prop:recover_V_from_t_hyb}
	For each $s\in\mathcal S$, suppose $\mathbf t^s\in\mathbb C^{N_{\mathrm{RF}}}$ is given. A feasible collection
	$\{\mathbf V_k^s\}_{k=1}^K$ is any one satisfying
	\[
	\sum_{k=1}^K \mathbf V_k^s[m,:]\boldsymbol\omega_k^s=\mathbf t^s[m],
	\qquad \forall m=1,\ldots,N_{\mathrm{RF}}.
	\]
	Let $
	\bar{\boldsymbol\omega}^{\,s}
	\triangleq
	\big[(\boldsymbol\omega_1^s)^T~\cdots~(\boldsymbol\omega_K^s)^T\big]^T.$
	If $\|\bar{\boldsymbol\omega}^{\,s}\|_2^2>0$, then the minimum-Frobenius-norm feasible recovery is
\begin{align}
	\mathbf V_k^s[m,:]
=
\frac{\mathbf t^s[m]}{\sum_{j=1}^K\|\boldsymbol\omega_j^s\|_2^2}\,(\boldsymbol\omega_k^s)^H,
\qquad \forall m,\ k.\label{eq:V_update}
\end{align}
\end{proposition}

\begin{proof}
	Fix $s\in\mathcal S$. For each RF-chain index $m$, define the concatenated row $
	\bar{\mathbf v}^{\,s}[m,:]
	\triangleq
	\big[\mathbf V_1^s[m,:]\ \cdots\ \mathbf V_K^s[m,:]\big].$
	Then, the feasibility condition becomes $\bar{\mathbf v}^{\,s}[m,:]\bar{\boldsymbol\omega}^{\,s}
	=
	\sum_{k=1}^K \mathbf V_k^s[m,:]\boldsymbol\omega_k^s
	=
	\mathbf t^s[m],
	 \forall m,$
	so the recovery problem separates row by row. Moreover, $\sum_{k=1}^K\|\mathbf V_k^s\|_F^2
	=
	\sum_{m=1}^{N_{\mathrm{RF}}}\|\bar{\mathbf v}^{\,s}[m,:]\|_2^2$. Hence, minimizing the total Frobenius norm is equivalent to minimizing, for each $m$, the Euclidean norm of
	$\bar{\mathbf v}^{\,s}[m,:]$ subject to
	$\bar{\mathbf v}^{\,s}[m,:]\bar{\boldsymbol\omega}^{\,s}=\mathbf t^s[m]$.
	
	If $\|\bar{\boldsymbol\omega}^{\,s}\|_2^2>0$, the minimum-norm solution of this linear equation is given by the pseudoinverse $	\bar{\mathbf v}^{\,s}[m,:]
	=
	\mathbf t^s[m]\frac{(\bar{\boldsymbol\omega}^{\,s})^H}{\|\bar{\boldsymbol\omega}^{\,s}\|_2^2}$.
	Extracting the block corresponding to user $k$ yields \eqref{eq:V_update}.
\end{proof}
Once the problem for each OFDM symbol is solved, we denote the solution by $\mathbf t^s(\boldsymbol\omega^{(b)})$.
\begin{proposition}
	\label{prop:admm_convergence_hyb}
	Assume that $\eta_w>0$, $\eta_t>0$, and that the $N_t$ antennas are partitioned into
	$N_{\mathrm{RF}}$ subarrays of identical size.
	Let $\mathcal R(\mathbf Z)\triangleq \begin{bmatrix}\Re\{\operatorname{vec}(\mathbf Z)\}\\ \Im\{\operatorname{vec}(\mathbf Z)\}\end{bmatrix}$, and define
	$\boldsymbol{\iota}_Q=\mathcal R(\mathbf Q)$, $\boldsymbol{\iota}_X=\mathcal R(\mathbf X)$, $\boldsymbol{\iota}_W=\mathcal R(\mathbf W)$, and $\boldsymbol{\iota}_T=\mathcal R(\mathbf T)$.
	Let $\mathbf A_{\mathcal R}$ and $\mathbf F_{\mathcal R}$ be the unique real matrices satisfying
	$\mathcal R(\mathfrak A(\mathbf W))=\mathbf A_{\mathcal R}\boldsymbol{\iota}_W$ and $\mathcal R(\mathfrak F(\mathbf W))=\mathbf F_{\mathcal R}\boldsymbol{\iota}_W$.
	Also, let $\mathfrak S:\mathbb C^{N_{\mathrm{RF}}\times S}\to\mathbb C^{N_t\times S}$ be defined by $[\mathfrak S(\mathbf T)]_{a,:}\triangleq \mathbf T[m_a,:]$,
	$a=1,\ldots,N_t$, and let $\mathbf P\in\{0,1\}^{N_t\times N_{\mathrm{RF}}}$ satisfy $\mathbf P[a,m]=1$ if $m=m_a$ and $\mathbf P[a,m]=0$ otherwise.
	Then, $\mathfrak S(\mathbf T)=\mathbf P\mathbf T$,
	$\operatorname{vec}(\mathfrak S(\mathbf T))=(\mathbf I_S\otimes \mathbf P)\operatorname{vec}(\mathbf T)$, and
	$\mathbf S_{\mathcal R}\triangleq \begin{bmatrix}\mathbf I_S\otimes \mathbf P & \mathbf 0\\ \mathbf 0 & \mathbf I_S\otimes \mathbf P\end{bmatrix}$
	satisfies $\mathcal R(\mathfrak S(\mathbf T))=\mathbf S_{\mathcal R}\boldsymbol{\iota}_T$. If $\rho$ satisfies \cite[Eq.~(3.37)]{tao2018convergence} for the
	real-valued problem
	\[
	\begin{aligned}
		\min \quad
		&\theta_1(\mathcal R^{-1}(\boldsymbol{\iota}_Q))+\theta_2(\mathcal R^{-1}(\boldsymbol{\iota}_X))
		+\theta_3(\mathcal R^{-1}(\boldsymbol{\iota}_W))\\&\quad+\theta_4(\mathcal R^{-1}(\boldsymbol{\iota}_T))\\
		\text{s.t.}\quad
		&
		\left[\begin{smallmatrix}
			\mathbf I_{2N_tG} & \mathbf 0 & -\mathbf A_{\mathcal R} & \mathbf 0\\
			\mathbf 0 & \mathbf I_{2N_t\ell S} & -\mathbf F_{\mathcal R} & \mathbf 0\\
			\mathbf 0 & \mathbf 0 & \mathbf I_{2N_tS} & -\mathbf S_{\mathcal R}
		\end{smallmatrix}\right]
		\left[\begin{smallmatrix}
			\boldsymbol{\iota}_Q\\
			\boldsymbol{\iota}_X\\
			\boldsymbol{\iota}_W\\
			\boldsymbol{\iota}_T
		\end{smallmatrix}\right]
		=\mathbf 0,
	\end{aligned}
	\]
	then, the ADMM iterates in \eqref{eq:admm_4block_updates_hyb} converge to a primal--dual solution of \eqref{opt:inner4block_hyb}. Consequently,
	the primal limit $(\mathbf Q^\star,\mathbf X^\star,\mathbf W^\star,\mathbf T^\star)$ is a global optimal solution of
	\eqref{opt:inner4block_hyb} and satisfies its KKT conditions.
	
	Moreover, if $\{(\eta_w^\nu,\eta_t^\nu)\}_\nu$ is any sequence with
	$(\eta_w^\nu,\eta_t^\nu)\to(0,0)$ and $(\eta_w^\nu,\eta_t^\nu)>\mathbf 0$,
	and $(\mathbf Q^\nu,\mathbf X^\nu,\mathbf W^\nu,\mathbf T^\nu)$ denotes a
	global optimal solution of \eqref{opt:inner4block_hyb} with parameters
	$(\eta_w^\nu,\eta_t^\nu)$, then every accumulation point of
	$\{(\mathbf Q^\nu,\mathbf X^\nu,\mathbf W^\nu,\mathbf T^\nu)\}_\nu$ is a
	global optimal solution of the unregularized inner problem.
\end{proposition}

\begin{proof}
	The map $\mathcal R$ is linear and bijective, so \eqref{opt:inner4block_hyb} and its real-valued counterpart are equivalent. We verify Assumptions~2.1--2.2 of \cite{tao2018convergence} for the realified four-block problem. The functions $\theta_1=\delta_{\mathcal M}$ and $\theta_2=\delta_{\mathcal C}$ are closed proper convex, while $\theta_3$ and $\theta_4$ are strongly convex with moduli $\eta_w$ and $\eta_t$, respectively. Each block column of the constraint matrix has full column rank: the first two by construction, the third because $\mathbf A_{\mathcal R}^T\mathbf A_{\mathcal R}+\mathbf F_{\mathcal R}^T\mathbf F_{\mathcal R}+\mathbf I_{2N_tS}\succ\mathbf 0$, and the fourth because $\mathbf S_{\mathcal R}^T\mathbf S_{\mathcal R}=\frac{N_t}{N_{\mathrm{RF}}}\mathbf I_{2N_{\mathrm{RF}}S}\succ\mathbf 0$ since each subarray has $N_t/N_{\mathrm{RF}}$ antennas. Since $r_j>0$, $\chi>0$, and $P^s>0$ for all $j,s$, the origin is feasible and lies in the relative interior of the product of the four block domains, confirming Assumption~2.2. The feasible set is compact because boundedness of $\mathbf T\in\mathcal P$ propagates through the linear equalities to $\mathbf W$, $\mathbf Q$, and $\mathbf X$, while closedness follows from closedness of $\mathcal M$, $\mathcal C$, and $\mathcal P$ together with the linear equality constraints. Since the objective is proper and lower semicontinuous, it attains its minimum on the compact feasible set, so the solution set is nonempty. Hence, if $\rho$ satisfies \cite[Eq.~(3.37)]{tao2018convergence}, then \cite[Thm.~3.1]{tao2018convergence} yields convergence of the real ADMM iterates to a primal--dual solution of the realified problem. Mapping back through $\mathcal R^{-1}$ gives convergence of \eqref{eq:admm_4block_updates_hyb} to a primal--dual solution of \eqref{opt:inner4block_hyb}. Since \eqref{opt:inner4block_hyb} is convex, the primal limit is globally optimal and satisfies the KKT conditions.
	
For the regularization-vanishing claim, let $\boldsymbol{\alpha}\triangleq(\mathbf Q,\mathbf X,\mathbf W,\mathbf T)$, let $\mathcal F$ denote the common feasible set of the regularized and unregularized inner problems, and let $F_0$ be the unregularized objective. Since $\mathcal F$ is compact, the sequence of regularized minimizers $\{\boldsymbol{\alpha}^\nu\}$ admits accumulation points; let $\bar{\boldsymbol{\alpha}}$ be one such point, attained along a subsequence $\nu'$. Because $\mathcal F$ is closed, $\bar{\boldsymbol{\alpha}}\in\mathcal F$. Regularized optimality gives $F_0(\boldsymbol{\alpha}^{\nu'})+R^{\nu'}(\boldsymbol{\alpha}^{\nu'})\le F_0(\boldsymbol{\alpha})+R^{\nu'}(\boldsymbol{\alpha})$ for all $\boldsymbol{\alpha}\in\mathcal F$, where $R^{\nu'}(\boldsymbol{\alpha})\triangleq \frac{\eta_w^{\nu'}}{2}\|\mathbf W\|_F^2+\frac{\eta_t^{\nu'}}{2}\|\mathbf T\|_F^2$. Since $R^{\nu'}(\boldsymbol{\alpha}^{\nu'})\to0$ because $\{\boldsymbol{\alpha}^{\nu'}\}$ is bounded, and $R^{\nu'}(\boldsymbol{\alpha})\to0$ for each fixed $\boldsymbol{\alpha}\in\mathcal F$, taking $\limsup_{\nu'\to\infty}$ yields $\limsup_{\nu'\to\infty}F_0(\boldsymbol{\alpha}^{\nu'})\le F_0(\boldsymbol{\alpha})$ for all $\boldsymbol{\alpha}\in\mathcal F$. By lower semicontinuity of $F_0$, we obtain $F_0(\bar{\boldsymbol{\alpha}})\le \liminf_{\nu'\to\infty}F_0(\boldsymbol{\alpha}^{\nu'})\le \limsup_{\nu'\to\infty}F_0(\boldsymbol{\alpha}^{\nu'})\le F_0(\boldsymbol{\alpha})$ for all $\boldsymbol{\alpha}\in\mathcal F$. Therefore, $\bar{\boldsymbol{\alpha}}$ is a global optimal solution of the unregularized inner problem.
\end{proof}

\begin{algorithm}[t!]
	\caption{The Proposed ADMM Algorithm}
	\label{al:admm}
	
	Initialize $\tau\gets 0$, feasible $\mathbf T^0$, $\mathbf W^0$, $\mathbf X^0$, $\mathbf Q^0$, and dual matrices
	$\boldsymbol\Lambda_q^0$, $\boldsymbol\Lambda_x^0$, $\boldsymbol\Lambda_w^0$\;
	
	\While{stopping criterion is not satisfied}{
		Update $\mathbf Q^{\tau+1}$ via \eqref{eq:q_closed_hyb} in parallel $\forall a$\;
		
		Update $\mathbf X^{\tau+1}$ via \eqref{eq:x_closed_hyb} in parallel $\forall a$\;
		
		Update $\mathbf W^{\tau+1}$ via \eqref{eq:w_normal_per_antenna_hyb} in parallel $\forall a$\;
		
		Update $\mathbf T^{\tau+1}$ by solving \eqref{opt:t_sub_hyb} in parallel $s\in\mathcal S$\;
		
		Update $\boldsymbol\Lambda_q^{\tau+1}$, $\boldsymbol\Lambda_x^{\tau+1}$, and $\boldsymbol\Lambda_w^{\tau+1}$
		according to \eqref{eq:admm_4block_updates_hyb}$\,$\;
		
		$\tau\gets\tau+1$\;
	}
	
	Recover $\{\mathbf V_k^{s,\tau}\}_{k=1,\ldots,K;\,s\in\mathcal S}$ from $\{\mathbf t^{s,\tau}\}_{s\in\mathcal S}$
	using Proposition~\ref{prop:recover_V_from_t_hyb} in parallel over $s\in\mathcal S$\;
	
	\textbf{Output:} $\{\mathbf V_k^{s,\tau}\}_{k=1,\ldots,K;\,s\in\mathcal S}$ and $\{\mathbf t^{s,\tau}\}_{s\in\mathcal S}$\;
	
\end{algorithm}
	
\subsection{The Subproblem with Respect to $\mathbf{V}_{\text{RF}}$}\label{sec:analogp}

Like the digital precoder, $\mathbf{V}_{\text{RF}}$ is optimized independently for each OFDM symbol realization. We note that each row of $\mathbf{V}_{\text{RF}}$ contains exactly one non-zero entry (the PS coefficient of that antenna),
and this entry has unit magnitude. Let $\bar{\mathbf{v}}_{\text{PS}}\triangleq \text{vec}(\mathbf{V}_{\text{RF}})$ and
let \texttt{idx} denote the indices of the non-zero entries of $\bar{\mathbf{v}}_{\text{PS}}$ (one index per row of $\mathbf{V}_{\text{RF}}$).
We then define
\begin{align}
	\mathbf{v}_{\text{PS}}=\bar{\mathbf{v}}_{\text{PS}}[\texttt{idx}]\in \mathbb{C}^{N_t},
	\qquad |\mathbf{v}_{\text{PS}}[a]|=1,\ \forall a.
	\label{eq:vps_from_idx}
\end{align}

\medskip
\noindent We consider the MSE matrix expression in \eqref{eq:E_det_t_expand} and ignore all terms that are independent of $\mathbf{V}_{\text{RF}}$. We obtain the
$\mathbf{V}_{\text{RF}}$-dependent part of the MSE matrix:
\begin{align}
	\mathbf{E}_{k}^{s,\text{simplified}}
	&=
	(\mathbf U_k^{s})^{H}\mathbf{U}_{\text{RF},k}^{H}
	\mathbf{H}_k^s\mathbf{V}_{\text{RF}}\mathbf t^s(\mathbf t^s)^H
	\mathbf{V}_{\text{RF}}^{H}(\mathbf{H}_k^s)^{H}
	\mathbf{U}_{\text{RF},k}\mathbf{U}_k^s \nonumber\\
	&\quad
	-(\mathbf U_k^{s})^{H}\mathbf{U}_{\text{RF},k}^{H}
	\mathbf{H}_k^s\mathbf{V}_{\text{RF}}\mathbf t^s(\boldsymbol{\omega}_k^{s})^{H}
	\nonumber\\
	&-\Big(
	(\mathbf U_k^{s})^{H}\mathbf{U}_{\text{RF},k}^{H}
	\mathbf{H}_k^s\mathbf{V}_{\text{RF}}\mathbf t^s(\boldsymbol{\omega}_k^{s})^{H}
	\Big)^{H}.
	\label{eq:Ek_simplified}
\end{align}
To optimize $\mathbf{V}_{\text{RF}}$, we minimize $\sum_{s\in \mathcal{S}} \sum_{k=1}^{K}\tr(\mathbf{E}_{k}^{s,\text{simplified}}).$ 
From \eqref{eq:Ek_simplified}, the linear dependence on $\mathbf{V}_{\text{RF}}$ appears in the two cross terms, which contribute
\begin{align}
	-2\Re\!\Bigg(
	\sum_{s\in \mathcal{S}} \sum_{k=1}^{K}
	\tr\!\Big(
	(\mathbf U_k^{s})^{H}\mathbf{U}_{\text{RF},k}^{H}
	\mathbf{H}_k^s\mathbf{V}_{\text{RF}}\mathbf t^s(\boldsymbol{\omega}_k^{s})^{H}
	\Big)
	\Bigg),
	\nonumber
\end{align}
to the objective function. Using the cyclicity of the trace,
\begin{align}
	&	\tr\!\Big(
	(\mathbf U_k^{s})^{H}\mathbf{U}_{\text{RF},k}^{H}
	\mathbf{H}_k^s\mathbf{V}_{\text{RF}}\mathbf t^s(\boldsymbol{\omega}_k^{s})^{H}
	\Big)
	\nonumber\\
	&	=
	\tr\!\Big(
	\mathbf{V}_{\text{RF}}\mathbf t^s(\boldsymbol{\omega}_k^{s})^{H}
	(\mathbf U_k^{s})^{H}\mathbf{U}_{\text{RF},k}^{H}\mathbf{H}_k^s
	\Big).
	\label{eq:linear_cyclic}
\end{align}
Define $\mathbf{B}
	\triangleq
	\sum_{s\in \mathcal{S}} \sum_{k=1}^{K}
	\Big(\mathbf t^s(\boldsymbol{\omega}_k^{s})^{H}
	(\mathbf U_k^{s})^{H}\mathbf{U}_{\text{RF},k}^{H}\mathbf{H}_k^s\Big)^{H}.$
Summing \eqref{eq:linear_cyclic} over $(k,s)$ gives
\begin{align}
	\sum_{s\in \mathcal{S}} \sum_{k=1}^{K}
	\tr\!\Big(
	(\mathbf U_k^{s})^{H}\mathbf{U}_{\text{RF},k}^{H}
	\mathbf{H}_k^s\mathbf{V}_{\text{RF}}\mathbf t^s(\boldsymbol{\omega}_k^{s})^{H}
	\Big)
	=
	\tr\!\big(\mathbf{B}^H\mathbf{V}_{\text{RF}}\big).
	\nonumber
\end{align}
Using $\tr(\mathbf{B}^H\mathbf{V}_{\text{RF}})=\text{vec}(\mathbf{B})^H\text{vec}(\mathbf{V}_{\text{RF}})$, let
$\bar{\mathbf{u}}_{\text{PS}}\triangleq \text{vec}(\mathbf{B})$ and $\bar{\mathbf{v}}_{\text{PS}}\triangleq \text{vec}(\mathbf{V}_{\text{RF}})$.
Restricting $\bar{\mathbf{v}}_{\text{PS}}$ to its non-zero entries according to  \eqref{eq:vps_from_idx} and similarly defining $\mathbf{u}_{\text{PS}}=\bar{\mathbf{u}}_{\text{PS}}[\texttt{idx}]$, yields
\begin{align}
	\tr(\mathbf{B}^H\mathbf{V}_{\text{RF}})
	=\mathbf{u}_{\text{PS}}^{H}\mathbf{v}_{\text{PS}}.
	\label{eq:linear_uv}
\end{align}
Hence, the linear contribution to the objective is $-2\Re\!\big(\mathbf{u}_{\text{PS}}^{H}\mathbf{v}_{\text{PS}}\big)$.

Consider the quadratic term in \eqref{eq:Ek_simplified}. By the cyclicity of the trace, we have
\begin{align}
	&\tr\!\Big(
	(\mathbf U_k^{s})^{H}\mathbf{U}_{\text{RF},k}^{H}
	\mathbf{H}_k^s\mathbf{V}_{\text{RF}}\mathbf t^s(\mathbf t^s)^H
	\mathbf{V}_{\text{RF}}^{H}(\mathbf{H}_k^s)^{H}
	\mathbf{U}_{\text{RF},k}\mathbf{U}_k^s
	\Big)\nonumber\\
	&=
	\tr\!\Big(
	\mathbf{V}_{\text{RF}}\mathbf t^s(\mathbf t^s)^H\mathbf{V}_{\text{RF}}^{H}\,
	(\mathbf{H}_k^s)^{H}\mathbf{U}_{\text{RF},k}\mathbf{U}_k^s
	(\mathbf U_k^{s})^{H}\mathbf{U}_{\text{RF},k}^{H}\mathbf{H}_k^s
	\Big).
	\label{eq:quad_cyclic}
\end{align}
We use the standard identity
\begin{align}
	\tr\!\big(\mathbf{V}\mathbf{B}\mathbf{V}^H\mathbf A\big)
	=
	\text{vec}(\mathbf{V})^H\big(\mathbf{B}^T\otimes \mathbf A\big)\text{vec}(\mathbf{V}),
	\label{eq:tr_vec_kron}
\end{align}
which holds for general $\mathbf A$, $\mathbf{B}$, and $\mathbf{V}$. Applying \eqref{eq:tr_vec_kron} to \eqref{eq:quad_cyclic} with
$\mathbf{V}=\mathbf{V}_{\text{RF}}$,
$\mathbf{B}=\mathbf t^s(\mathbf t^s)^H$, and
$\mathbf A=(\mathbf{H}_k^s)^{H}\mathbf{U}_{\text{RF},k}\mathbf{U}_k^s
(\mathbf U_k^{s})^{H}\mathbf{U}_{\text{RF},k}^{H}\mathbf{H}_k^s$, we obtain the following.
\begin{align}
	&\tr\!\Big(
	(\mathbf U_k^{s})^{H}\mathbf{U}_{\text{RF},k}^{H}
	\mathbf{H}_k^s\mathbf{V}_{\text{RF}}\mathbf t^s(\mathbf t^s)^H
	\mathbf{V}_{\text{RF}}^{H}(\mathbf{H}_k^s)^{H}
	\mathbf{U}_{\text{RF},k}\mathbf{U}_k^s
	\Big)	\label{eq:quad_vec_form}\\
	&=
	\bar{\mathbf{v}}_{\text{PS}}^{H}
	\Big(
	\big(\mathbf t^s(\mathbf t^s)^H\big)^T
	\otimes
	(\mathbf{H}_k^s)^{H}\mathbf{U}_{\text{RF},k}\mathbf{U}_k^s
	(\mathbf U_k^{s})^{H}\mathbf{U}_{\text{RF},k}^{H}\mathbf{H}_k^s
	\Big)
	\bar{\mathbf{v}}_{\text{PS}}.\nonumber
\end{align}
Summing \eqref{eq:quad_vec_form} over $s\in\mathcal{S}$ and $k=1,\ldots,K$, and then restricting to the feasible (non-zero)
entries via \texttt{idx}, we obtain a quadratic form in $\mathbf{v}_{\text{PS}}$:
\begin{align}
	&\sum_{s\in\mathcal{S}}\sum_{k=1}^{K}
	\tr\!\Big(
	(\mathbf U_k^{s})^{H}\mathbf{U}_{\text{RF},k}^{H}
	\mathbf{H}_k^s\mathbf{V}_{\text{RF}}\mathbf t^s(\mathbf t^s)^H
	\mathbf{V}_{\text{RF}}^{H}(\mathbf{H}_k^s)^{H}
	\mathbf{U}_{\text{RF},k}\mathbf{U}_k^s
	\Big)\nonumber\\
	&=
	\bar{\mathbf v}_{\text{PS}}^{H}\,\mathbf M_{\text{PS}}\,\bar{\mathbf v}_{\text{PS}}
	=
	\mathbf v_{\text{PS}}^{H}\,\mathbf Q_{\text{PS}}\,\mathbf v_{\text{PS}},
	\label{eq:quad_vps}
\end{align}
where
\begin{align}
	&	\mathbf M_{\text{PS}}
	\triangleq
	\sum_{s\in\mathcal{S}}\sum_{k=1}^{K}
	\Big(
	\big(\mathbf t^s(\mathbf t^s)^H\big)^T\hspace{-.2cm}
	\otimes
	(\mathbf{H}_k^s)^{H}\mathbf{U}_{\text{RF},k}\mathbf{U}_k^s
	(\mathbf U_k^{s})^{H}\mathbf{U}_{\text{RF},k}^{H}\mathbf{H}_k^s
	\Big),\nonumber
	\\
	&	\mathbf Q_{\text{PS}}
	\triangleq
	\mathbf M_{\text{PS}}[\texttt{idx},\texttt{idx}].
	\nonumber
\end{align}
Since $\mathbf Q_{\mathrm{PS}}$ is Hermitian, the quadratic term is real-valued. Combining \eqref{eq:linear_uv} and \eqref{eq:quad_vps}, the $\mathbf{V}_{\text{RF}}$-dependent part of
the trace sum-MSE can be written as
\begin{align}
	f(\mathbf{v}_{\text{PS}})
	=
	\mathbf{v}_{\text{PS}}^{H}\mathbf{Q}_{\text{PS}}\mathbf{v}_{\text{PS}}
	-2\Re\!\Big(\mathbf{u}_{\text{PS}}^{H}\mathbf{v}_{\text{PS}}\Big).
	\label{eq:objps}
\end{align}
Fix $\mathbf{v}_{\text{PS}}[j]$ for $j\neq a$. Collecting only the terms that depend on $\mathbf{v}_{\text{PS}}[a]$ in \eqref{eq:objps} yields
(up to terms independent of $\mathbf{v}_{\text{PS}}[a]$):
\begin{align}
	&f(\mathbf{v}_{\text{PS}})
	=
	\text{const.}\label{eq:coord_collect}\\
	&
	+2\Re\!\left(
	\mathbf{v}_{\text{PS}}[a]^*
	\Big(
	\mathbf{Q}_{\text{PS}}[a,:]\mathbf{v}_{\text{PS}}
	-\mathbf{Q}_{\text{PS}}[a,a]\mathbf{v}_{\text{PS}}[a]
	-\mathbf{u}_{\text{PS}}[a]
	\Big)
	\right),
	\nonumber
\end{align}
Since $|\mathbf{v}_{\text{PS}}[a]|=1$, the term $\mathbf{Q}_{\text{PS}}[a,a]|\mathbf{v}_{\text{PS}}[a]|^2$ is constant in $\mathbf{v}_{\text{PS}}[a]$
and can be absorbed into the constant term. Using $|x+y|^2=|x|^2+|y|^2+2\Re(x^*y)$ with
\[
x=\mathbf{v}_{\text{PS}}[a],\qquad
y=\mathbf{Q}_{\text{PS}}[a,:]\mathbf{v}_{\text{PS}}
-\mathbf{Q}_{\text{PS}}[a,a]\mathbf{v}_{\text{PS}}[a]
-\mathbf{u}_{\text{PS}}[a].
\]
minimizing \eqref{eq:coord_collect} over $|\mathbf{v}_{\text{PS}}[a]|=1$ is equivalent to minimizing $|\mathbf{v}_{\text{PS}}[a]+y|^2$
over the unit circle. Therefore, the minimizer is the projection of $-y$ onto $|\mathbf{v}_{\text{PS}}[a]|=1$, yielding
\[
\mathbf{v}_{\text{PS}}[a]
=
-\frac{
	\mathbf{Q}_{\text{PS}}[a,:]\mathbf{v}_{\text{PS}}
	-\mathbf{Q}_{\text{PS}}[a,a]\mathbf{v}_{\text{PS}}[a]
	-\mathbf{u}_{\text{PS}}[a]
}{
	\left|
	\mathbf{Q}_{\text{PS}}[a,:]\mathbf{v}_{\text{PS}}
	-\mathbf{Q}_{\text{PS}}[a,a]\mathbf{v}_{\text{PS}}[a]
	-\mathbf{u}_{\text{PS}}[a]
	\right|
}.
\]
We use coordinate descent, iteratively updating each element $\mathbf{v}_{\text{PS}}[a]$ while fixing others, using current iteration values
for elements $1,\ldots,a-1$ and previous iteration values for $a+1,\ldots,N_t$. Each update maintains unit magnitude and yields a non-increasing
sequence of objective values $f(\mathbf{v}_{\text{PS}})$; under continuity of the objective and exact minimization of each coordinate subproblem on the unit-modulus set, every accumulation point of the iterates is
a stationary point \cite{bertsekas1999nonlinear}.

PS impairments introduce independent random phase errors at each antenna element, modeled as complex exponentials with Gaussian-distributed
phase deviations. Unlike common local oscillator phase noise affecting all antennas equally, PS errors are element-specific and occur when new precoding
matrices are applied, distorting the intended precoding phase pattern through uncorrelated random phase deviations \cite{9767793,9069262}.
Starting from \eqref{eq:objps}, we consider the expected objective under phase errors:
\begin{align}
	\mathbb{E}_{\mathbf e_{\text{error}}}\!\left[f(\mathbf v_{\text{PS}})\right]
	&=
	\mathbb{E}_{\mathbf e_{\text{error}}}\!\left[\mathbf v_{\text{PS}}^{H}\mathbf Q_{\text{PS}}\mathbf v_{\text{PS}}\right]
	-2\,\mathbb{E}_{\mathbf e_{\text{error}}}\!\left[\Re\!\big(\mathbf u_{\text{PS}}^{H}\mathbf v_{\text{PS}}\big)\right].
	\label{eq:exp_obj_start}
\end{align}
Let $\mathbf v_{\text{PS}}=\tilde{\mathbf v}_{\text{PS}}\odot \mathbf e_{\text{error}}$, where $\tilde{\mathbf v}_{\text{PS}}$ denotes the
intended phase-shift vector and
\begin{align}
	\mathbf e_{\text{error}}=
	\big[e^{\jmath\Delta\theta_1},e^{\jmath\Delta\theta_2},\ldots,e^{\jmath\Delta\theta_{N_t}}\big]^T,
	\qquad
	\Delta\theta_i\sim\mathcal N(0,\sigma_e^2),
	\nonumber
\end{align}
with independent $\{\Delta\theta_i\}$. Expanding $\mathbf v_{\text{PS}}^{H}\mathbf Q_{\text{PS}}\mathbf v_{\text{PS}}$ element-wise and substituting
$\mathbf v_{\text{PS}}=\tilde{\mathbf v}_{\text{PS}}\odot \mathbf e_{\text{error}}$ gives
\begin{small}
	\begin{align}
		\mathbb{E}_{\mathbf e_{\text{error}}}\!\left[\mathbf v_{\text{PS}}^{H}\mathbf Q_{\text{PS}}\mathbf v_{\text{PS}}\right]
		&=
		\sum_{i=1}^{N_t}\sum_{j=1}^{N_t} \mathbf Q_{\text{PS}}[i,j]\,
		\tilde{\mathbf v}_{\text{PS}}^{*}[i]\tilde{\mathbf v}_{\text{PS}}[j]\,
		\mathbb{E}\!\left[e^{\jmath(\Delta\theta_j-\Delta\theta_i)}\right].
		\label{eq:exp_quad_start}
	\end{align}
\end{small}
For $i=j$, $\mathbb{E}[e^{\jmath(\Delta\theta_i-\Delta\theta_i)}]=1$. For $i\neq j$, independence implies
\begin{align}
	\mathbb{E}\!\left[e^{\jmath(\Delta\theta_j-\Delta\theta_i)}\right]
	&=
	\mathbb{E}\!\left[e^{\jmath\Delta\theta_j}\right]\,
	\mathbb{E}\!\left[e^{-\jmath\Delta\theta_i}\right]\nonumber\\
	&
	=
	e^{-\sigma_e^2/2}\,e^{-\sigma_e^2/2}
	=
	e^{-\sigma_e^2},
	\qquad i\neq j,
	\label{eq:exp_diff_gauss}
\end{align}
where we used $\mathbb{E}[e^{\jmath\Delta\theta_i}]=e^{-\sigma_e^2/2}$ for $\Delta\theta_i\sim\mathcal N(0,\sigma_e^2)$.
Substituting \eqref{eq:exp_diff_gauss} into \eqref{eq:exp_quad_start} and separating diagonal and off-diagonal terms yields
\begin{align}
	&\mathbb{E}_{\mathbf e_{\text{error}}}\!\left[\mathbf v_{\text{PS}}^{H}\mathbf Q_{\text{PS}}\mathbf v_{\text{PS}}\right]
	=
	\sum_{i=1}^{N_t}\mathbf Q_{\text{PS}}[i,i]\;|\tilde{\mathbf v}_{\text{PS}}[i]|^2
		\nonumber\\
	&+e^{-\sigma_e^2}\!\!\sum_{i=1}^{N_t}\sum_{\substack{j=1\\j\neq i}}^{N_t}
	\mathbf Q_{\text{PS}}[i,j]\;\tilde{\mathbf v}_{\text{PS}}^{*}[i]\tilde{\mathbf v}_{\text{PS}}[j]=
	\tilde{\mathbf v}_{\text{PS}}^{H}\big(\mathbf Q_{\text{PS}}\odot \mathbf I_{N_t}\big)\tilde{\mathbf v}_{\text{PS}}   \nonumber\\
	&
	+e^{-\sigma_e^2}\tilde{\mathbf v}_{\text{PS}}^{H}\big(\mathbf Q_{\text{PS}}\odot (\mathbf 1_{N_t\times N_t}-\mathbf I_{N_t})\big)\tilde{\mathbf v}_{\text{PS}}.\label{eq:exp_quad}
\end{align}
Using the linearity of expectation and the fact that $\Re(\cdot)$ is a real-linear operator, we write
\begin{align}
	\mathbb{E}_{\mathbf e_{\text{error}}}\!\left[\Re\!\big(\mathbf u_{\text{PS}}^{H}\mathbf v_{\text{PS}}\big)\right]
	=
	\Re\!\left(\mathbf u_{\text{PS}}^{H}\,\mathbb{E}_{\mathbf e_{\text{error}}}[\mathbf v_{\text{PS}}]\right).
	\label{eq:exp_lin_start}
\end{align}
Since $\mathbf v_{\text{PS}}=\tilde{\mathbf v}_{\text{PS}}\odot \mathbf e_{\text{error}}$ and the entries of $\mathbf e_{\text{error}}$
are i.i.d., we have $\mathbb{E}[\mathbf e_{\text{error}}]=e^{-\sigma_e^2/2}\mathbf 1_{N_t\times 1}$ and thus
$\mathbb{E}[\mathbf v_{\text{PS}}]=e^{-\sigma_e^2/2}\tilde{\mathbf v}_{\text{PS}}$. Substituting into \eqref{eq:exp_lin_start} yields
\begin{align}
	\mathbb{E}_{\mathbf e_{\text{error}}}\!\left[\Re\!\big(\mathbf u_{\text{PS}}^{H}\mathbf v_{\text{PS}}\big)\right]
	=
	e^{-\sigma_e^2/2}\Re\!\big(\mathbf u_{\text{PS}}^{H}\tilde{\mathbf v}_{\text{PS}}\big).
	\label{eq:exp_lin}
\end{align}
Substituting \eqref{eq:exp_quad} and \eqref{eq:exp_lin} into \eqref{eq:exp_obj_start} yields
\begin{align}
	\mathbb{E}_{\mathbf e_{\text{error}}}\!\left[f(\mathbf v_{\text{PS}})\right]
	=
	\tilde{\mathbf v}_{\text{PS}}^{H}\tilde{\mathbf Q}_{\text{PS}}\tilde{\mathbf v}_{\text{PS}}
	-2e^{-\sigma_e^2/2}\Re\!\big(\mathbf u_{\text{PS}}^{H}\tilde{\mathbf v}_{\text{PS}}\big),
	\label{eq:exp_obj_final}
\end{align}
where $
	\tilde{\mathbf Q}_{\text{PS}}
	\triangleq
	\mathbf Q_{\text{PS}}\odot \mathbf I_{N_t}
	+e^{-\sigma_e^2}\big(\mathbf Q_{\text{PS}}\odot (\mathbf 1_{N_t\times N_t}-\mathbf I_{N_t})\big).$
Applying the same coordinate-descent argument used for \eqref{eq:objps} to \eqref{eq:exp_obj_final} gives, for $a=1,\ldots,N_t$,
\[
\tilde{\mathbf v}_{\text{PS}}[a]
=
-\frac{
	\tilde{\mathbf Q}_{\text{PS}}[a,:]\tilde{\mathbf v}_{\text{PS}}
	-\tilde{\mathbf Q}_{\text{PS}}[a,a]\tilde{\mathbf v}_{\text{PS}}[a]
	-e^{-\sigma_e^2/2}\mathbf u_{\text{PS}}[a]
}{
	\left|
	\tilde{\mathbf Q}_{\text{PS}}[a,:]\tilde{\mathbf v}_{\text{PS}}
	-\tilde{\mathbf Q}_{\text{PS}}[a,a]\tilde{\mathbf v}_{\text{PS}}[a]
	-e^{-\sigma_e^2/2}\mathbf u_{\text{PS}}[a]
	\right|
}.
\]
The same derivation can be repeated for other phase-error distributions by replacing $\mathbb{E}[e^{\jmath\Delta\theta}]$
in \eqref{eq:exp_diff_gauss} accordingly.

Once the problem for each OFDM symbol is solved, we denote the solution by $\mathbf V_{\mathrm{RF}}^{(b)}$.

\subsection{The Subproblem with Respect to $\mathbf U_k^s$}
For fixed transmit-side variables and fixed analog combiner $\mathbf U_{\mathrm{RF},k}$, the $\mathbf U_k^s$-subproblem is unconstrained and quadratic in $\mathbf U_k^s$. Because $\mathbf V_{\mathrm{RF}}$ is symbol-dependent (i.e., realization-specific), we define the per-realization effective channel $\bar{\mathbf H}_k^{s(b)}\triangleq \mathbf U_{\mathrm{RF},k}^H\mathbf H_k^s\mathbf V_{\mathrm{RF}}^{(b)}$. Setting $\partial J/\partial \mathbf U_k^{s*}=\mathbf 0$ yields the linear MMSE combiner
\begin{eqnarray}
	\mathbf U_k^s=&\hspace{-.3cm}\Big(\frac{1}{B}\sum_{b=1}^{B}\bar{\mathbf H}_k^{s(b)}\mathbf R_{t^st^s}^{(b)}\bar{\mathbf H}_k^{s(b)H}+\sigma_{\mathrm{noise},k}^{s\,2}\mathbf U_{\mathrm{RF},k}^H\mathbf U_{\mathrm{RF},k}\Big)^{-1}\nonumber\\
	&\times\frac{1}{B}\sum_{b=1}^{B}\bar{\mathbf H}_k^{s(b)}\mathbf R_{t^s\omega_k^s}^{(b)},\label{eq:U}
\end{eqnarray}
where $\mathbf R_{t^st^s}^{(b)}\triangleq \mathbf t^s(\boldsymbol{\omega}^{(b)})\mathbf t^s(\boldsymbol{\omega}^{(b)})^H$ and $\mathbf R_{t^s\omega_k^s}^{(b)}\triangleq \mathbf t^s(\boldsymbol{\omega}^{(b)})(\boldsymbol{\omega}_k^{s(b)})^H$, with each $\mathbf t^s (\boldsymbol{\omega}^{(b)})$ produced by the inner ADMM of Section~\ref{subsec:admm_blocks_revised_hyb}. Since $\mathbf V_{\mathrm{RF}}^{(b)}$ is folded into $\bar{\mathbf H}_k^{s(b)}$ inside the sum, it cannot be pulled out of the sample averages. The symbol decision is $\mathcal Q(\mathbf U_k^{sH}\mathbf U_{\mathrm{RF},k}^H\mathbf y_k^s)$, with $\mathcal Q$ the componentwise QAM projection.

\subsection{The Subproblem with Respect to Analog Combining $\mathbf{U}_{\text{RF},k}$}\label{sec:updateU}

Because the analog combiner is part of the symbol-agnostic receiver, we update it by minimizing the symbol-averaged cost of user $k$. In other words, we minimize $\sum_{s\in\mathcal S}\mathbb E_{\boldsymbol\omega}[\tr(\mathbf E_k^s)]$ with respect to
$\mathbf{U}_{\text{RF},k}$ while fixing all other variables. Starting from the expanded MSE expression in \eqref{eq:E_det_t_expand} and keeping only the terms that depend on
$\mathbf{U}_{\text{RF},k}$, the $\mathbf{U}_{\text{RF},k}$-dependent objective becomes
\begin{align}
	&\sum_{s\in\mathcal S}\mathbb E_{\boldsymbol\omega}\!\big[\tr(\mathbf E_k^s)\big]
	=
	-2\Re\!\Bigg(
	\sum_{s\in\mathcal S}
	\tr\!\Big(
	\mathbf{U}_{\text{RF},k}^{H}\mathbf{H}_k^s
	\mathbf R_{t^s\omega_k^s}(\mathbf U_k^{s})^{H}
	\Big)
	\Bigg)\nonumber\\
	&+
	\sum_{s\in\mathcal S}
	\tr\!\Big(
	(\mathbf U_k^{s})^{H}\mathbf{U}_{\text{RF},k}^{H}\bar{\mathbf D}_k^s\mathbf{U}_{\text{RF},k}\mathbf{U}_k^s
	\Big)\nonumber\\
	&
	+\text{(terms independent of $\mathbf{U}_{\text{RF},k}$)},
	\label{eq:simpcost}
\end{align}
where
$\bar{\mathbf D}_k^s
\triangleq
\mathbf{H}_k^s\mathbf R_{t^st^s}
(\mathbf{H}_k^s)^{H}
+\sigma_{\text{noise},k}^{s\,2}\mathbf{I}_{N_r},$
with the antenna-domain sample averages
$\mathbf R_{t^st^s}\triangleq \frac{1}{B}\sum_{b=1}^{B}\mathbf V_{\text{RF}}^{(b)}\mathbf t^s(\boldsymbol\omega^{(b)})\mathbf t^s(\boldsymbol\omega^{(b)})^H\mathbf V_{\text{RF}}^{(b)H}$ and $\mathbf R_{t^s\omega_k^s}\triangleq \frac{1}{B}\sum_{b=1}^{B}\mathbf V_{\text{RF}}^{(b)}\mathbf t^s(\boldsymbol\omega^{(b)})(\boldsymbol\omega_k^{s(b)})^H$. Since $\mathbf V_{\text{RF}}^{(b)}$ is realization-specific, it cannot be pulled out of the sample averages.

To update a generic feasible entry $(a,m)\in\mathcal E_k$ of $\mathbf U_{\text{RF},k}$ under the unit-modulus constraint $|\mathbf U_{\text{RF},k}[a,m]|=1$, we fix all other entries and isolate the dependence on this single variable.
Let $\mathbf{O}_{a,m}$ denote the matrix of the same size as $\mathbf{U}_{\text{RF},k}$ whose only non-zero entry equals $1$
at position $(a,m)$. Write
\begin{align}
	\mathbf{U}_{\text{RF},k}
	=
	\mathbf{U}_{\text{RF},k}^{(a,m)}+\mathbf{U}_{\text{RF},k}[a,m]\mathbf{O}_{a,m},
	\label{eq:U_decomp}
\end{align}
where $\mathbf{U}_{\text{RF},k}^{(a,m)}$ equals $\mathbf{U}_{\text{RF},k}$ with its $(a,m)^{\text{th}}$ entry set to zero. This way, $\mathbf{U}_{\text{RF},k}^{(a,m)}$ stays fixed while we optimize only $\mathbf{U}_{\text{RF},k}[a,m]$. Using the linearity and cyclicity of the trace, the linear part in \eqref{eq:simpcost} becomes
\begin{align}
	&\sum_{s\in\mathcal S}
	\tr\!\Big(
	\mathbf{U}_{\text{RF},k}^{H}\mathbf{H}_k^s
	\mathbf R_{t^s\omega_k^s}(\mathbf U_k^{s})^{H}
	\Big)\nonumber\\
	&=
	\sum_{s\in\mathcal S}
	\tr\!\Big(
	\big(\mathbf{U}_{\text{RF},k}^{(a,m)}\big)^{H}
	\mathbf{H}_k^s\mathbf R_{t^s\omega_k^s}(\mathbf U_k^{s})^{H}
	\Big)\nonumber\\
	&\quad+
	\mathbf{U}_{\text{RF},k}^*[a,m]\,
	\sum_{s\in\mathcal S}
	\tr\!\Big(
	\mathbf{O}_{a,m}^{H}
	\mathbf{H}_k^s\mathbf R_{t^s\omega_k^s}(\mathbf U_k^{s})^{H}
	\Big).\nonumber
\end{align}
Since $\tr(\mathbf{O}_{a,m}^H\mathbf{X})=\mathbf{X}[a,m]$ for any conformable matrix $\mathbf{X}$, define
$\beta_{1,a m}
\triangleq
\left[\sum_{s\in\mathcal S}
\Big(
\mathbf{H}_k^s\mathbf R_{t^s\omega_k^s}(\mathbf U_k^{s})^{H}
\Big)\right]_{a,m}.$
Then, the only part of the linear term depending on $\mathbf{U}_{\text{RF},k}[a,m]$ is
$\mathbf{U}_{\text{RF},k}^*[a,m]\beta_{1,a m}$, and the contribution of the two conjugate linear traces in \eqref{eq:simpcost}
is $-2\Re\!\big(\mathbf{U}_{\text{RF},k}^*[a,m]\beta_{1,a m}\big)$.

Using the cyclicity of the trace, the quadratic part in \eqref{eq:simpcost} can be rewritten as
\begin{align}
	\tr\!\Big(
	(\mathbf U_k^{s})^{H}\mathbf{U}_{\text{RF},k}^{H}\bar{\mathbf D}_k^s\mathbf{U}_{\text{RF},k}\mathbf{U}_k^s
	\Big)
	=
	\tr\!\Big(
	\mathbf{U}_{\text{RF},k}^{H}\bar{\mathbf D}_k^s\mathbf{U}_{\text{RF},k}\mathbf{U}_k^s(\mathbf U_k^{s})^{H}
	\Big).
	\label{eq:quad_cyclic_U}
\end{align}
Substituting \eqref{eq:U_decomp} into \eqref{eq:quad_cyclic_U} and expanding the product yields
\begin{align}
	&\tr\!\Big(
	\mathbf{U}_{\text{RF},k}^{H}\bar{\mathbf D}_k^s\mathbf{U}_{\text{RF},k}\mathbf{U}_k^s(\mathbf U_k^{s})^{H}
	\Big)\nonumber\\
	&=
	\tr\!\Big(
	\big(\mathbf{U}_{\text{RF},k}^{(a,m)}\big)^{H}\bar{\mathbf D}_k^s\mathbf{U}_{\text{RF},k}^{(a,m)}\mathbf{U}_k^s(\mathbf U_k^{s})^{H}
	\Big)\nonumber\\
	&\quad+
	\mathbf{U}_{\text{RF},k}^*[a,m]\,
	\tr\!\Big(
	\mathbf{O}_{a,m}^{H}\bar{\mathbf D}_k^s\mathbf{U}_{\text{RF},k}^{(a,m)}\mathbf{U}_k^s(\mathbf U_k^{s})^{H}
	\Big)\nonumber\\
	&\quad+
	\mathbf{U}_{\text{RF},k}[a,m]\,
	\tr\!\Big(
	\big(\mathbf{U}_{\text{RF},k}^{(a,m)}\big)^{H}\bar{\mathbf D}_k^s\mathbf{O}_{a,m}\mathbf{U}_k^s(\mathbf U_k^{s})^{H}
	\Big)\nonumber\\
	&\quad+
	|\mathbf{U}_{\text{RF},k}[a,m]|^2\,
	\tr\!\Big(
	\mathbf{O}_{a,m}^{H}\bar{\mathbf D}_k^s\mathbf{O}_{a,m}\mathbf{U}_k^s(\mathbf U_k^{s})^{H}
	\Big).
	\label{eq:quad_split}
\end{align}
Since each $\bar{\mathbf D}_k^s$ is Hermitian, the two cross terms in \eqref{eq:quad_split} are complex conjugates. The last term in \eqref{eq:quad_split} is constant with respect to $\mathbf{U}_{\text{RF},k}[a,m]$ because $|\mathbf{U}_{\text{RF},k}[a,m]|=1$.
Using again $\tr(\mathbf{O}_{a,m}^H\mathbf{Y})=\mathbf{Y}[a,m]$, define
$\beta_{2,a m}
\triangleq
\Big[\sum_{s\in\mathcal S}
\bar{\mathbf D}_k^s\mathbf{U}_{\text{RF},k}^{(a,m)}\mathbf{U}_k^s(\mathbf U_k^{s})^{H}
\Big]_{a,m}.$
Then, the $\mathbf{U}_{\text{RF},k}[a,m]$-dependent part of the quadratic term summed over $s$ is
$2\Re\!\big(\mathbf{U}_{\text{RF},k}^*[a,m]\beta_{2,a m}\big)$.

Combining the $\mathbf{U}_{\text{RF},k}[a,m]$-dependent parts of the linear and quadratic terms, the update of the
$(a,m)^{\text{th}}$ entry reduces to
\begin{align}
	\min_{|\mathbf{U}_{\text{RF},k}[a,m]|=1}\;
	2\Re\!\left\{\mathbf{U}_{\text{RF},k}^*[a,m]\big(\beta_{2,a m}-\beta_{1,a m}\big)\right\}.
	\label{eq:scalar_update_U}
\end{align}
Using $\Re\{z\}\ge -|z|$ and $|\mathbf{U}_{\text{RF},k}[a,m]|=1$, we have
$\Re\!\left\{\mathbf{U}_{\text{RF},k}^*[a,m]\big(\beta_{2,a m}-\beta_{1,a m}\big)\right\}
\ge
-\left|\beta_{2,a m}-\beta_{1,a m}\right|,$
and the lower bound is achieved when $\mathbf{U}_{\text{RF},k}^*[a,m](\beta_{2,a m}-\beta_{1,a m})$ is real and negative.
Therefore, 
\begin{eqnarray}
\mathbf{U}_{\text{RF},k}[a,m]
=
-\frac{\beta_{2,a m}-\beta_{1,a m}}{\left|\beta_{2,a m}-\beta_{1,a m}\right|}.\nonumber
\end{eqnarray}
	
We then run coordinate descent by sweeping over the feasible entries of $\mathbf{U}_{\text{RF},k}$ and updating one entry at a time while keeping the rest fixed. Each update preserves the unit-modulus constraint and does not increase the objective value. Under mild conditions (including continuity of the objective and exact minimization of each coordinate subproblem on the unit-modulus set) the objective values are non-increasing and every accumulation point of the iterates is a stationary point \cite{bertsekas1999nonlinear}.

In the presence of PS impairments at the user equipment that introduce random phase errors, one can use the same robust-optimization idea as
for the RF precoder at the transmitter. Specifically, we take the expectation of the scalar coordinate cost in \eqref{eq:scalar_update_U} with respect to the phase errors.
Assume
$\mathbf{U}_{\text{RF},k}= \tilde{\mathbf{U}}_{\text{RF},k} \odot \mathbf{E}_{\text{error}}$, where
$\mathbf{E}_{\text{error}}[a,m]=e^{\jmath e_{a,m}}$ and $e_{a,m}\sim\mathcal{N}(0,\sigma_e^2)$ are independent. Then,
$\mathbb{E}[e^{\jmath e_{a,m}}]=e^{-\sigma_e^2/2}$ and $\mathbb{E}[e^{-\jmath e_{a,m}}]=e^{-\sigma_e^2/2}$.

Since $\beta_{1,a m}$ is independent of $\mathbf{U}_{\text{RF},k}$, we have
$\mathbb{E}_{\mathbf{E}_{\text{error}}}[\mathbf{U}_{\text{RF},k}^*[a,m]\beta_{1,a m}]
=e^{-\sigma_e^2/2}\,\tilde{\mathbf{U}}_{\text{RF},k}^*[a,m]\beta_{1,a m}$.
Moreover, $\beta_{2,a m}$ is linear in the entries of $\mathbf{U}_{\text{RF},k}^{(a,m)}$, hence
$\mathbb{E}_{\mathbf{E}_{\text{error}}}[\beta_{2,a m}]=e^{-\sigma_e^2/2}\,\tilde{\beta}_{2,a m}$, where
$\tilde{\beta}_{2,a m}
\triangleq
\big[\sum_{s\in\mathcal S}
\bar{\mathbf D}_k^s\,\tilde{\mathbf{U}}_{\text{RF},k}^{(a,m)}\mathbf{U}_k^s(\mathbf U_k^{s})^{H}
\big]_{a,m}$.
Since $\mathbf{U}_{\text{RF},k}^*[a,m]$ depends only on the phase error $e_{a,m}$ while $\beta_{2,a m}$ depends only on the errors in
$\mathbf{U}_{\text{RF},k}^{(a,m)}$ (and is therefore independent of $e_{a,m}$), the product factorizes and picks up the factor $e^{-\sigma_e^2}$ as $\mathbb{E}_{\mathbf{E}_{\text{error}}}\!\big[
	\mathbf{U}_{\text{RF},k}^*[a,m]\beta_{2,a m}\big]
	=
	e^{-\sigma_e^2}\,\tilde{\mathbf{U}}_{\text{RF},k}^*[a,m]\tilde{\beta}_{2,a m}.$
Therefore, the expected scalar term in \eqref{eq:scalar_update_U} becomes
\begin{align}
	\min_{|\tilde{\mathbf{U}}_{\text{RF},k}[a,m]|=1}\hspace{-.3cm}
	2\Re\!\left\{\tilde{\mathbf{U}}_{\text{RF},k}^*[a,m]\Big(e^{-\sigma_e^2}\tilde{\beta}_{2,a m}-e^{-\sigma_e^2/2}\beta_{1,a m}\Big)\right\},
	\nonumber
\end{align}
and the closed-form update is
\begin{eqnarray}
\tilde{\mathbf{U}}_{\text{RF},k}[a,m]
=
-\frac{e^{-\sigma_e^2/2}\tilde{\beta}_{2,a m}-\beta_{1,a m}}
{\left|e^{-\sigma_e^2/2}\tilde{\beta}_{2,a m}-\beta_{1,a m}\right|},\nonumber
\end{eqnarray}
 with the convention that if the denominator is zero then $\tilde{\mathbf{U}}_{\text{RF},k}[a,m]$ can be kept unchanged.

	\subsection{The Proposed BCD-based Hybrid Precoding Approach}
In the proposed method, the variable blocks are updated cyclically. At each update, we fix all remaining blocks and optimize only the current
one. Specifically, $\mathbf{U}_k^s$ is updated using \eqref{eq:U}. The digital precoders $\{\mathbf{V}_k^s\}$ are updated via the ADMM-based
Algorithm~\ref{al:admm}. The PS matrix $\mathbf{V}_{\text{RF}}$ is updated using the procedure in Section~\ref{sec:analogp}. Finally, the
analog combiner $\mathbf{U}_{\text{RF},k}$ is updated using the mechanism in Section~\ref{sec:updateU}. 

\begin{figure*}[t!]
	\centering
	\subfloat[\label{fig:conv_bcd}]{
		\includegraphics[width=0.24\textwidth]{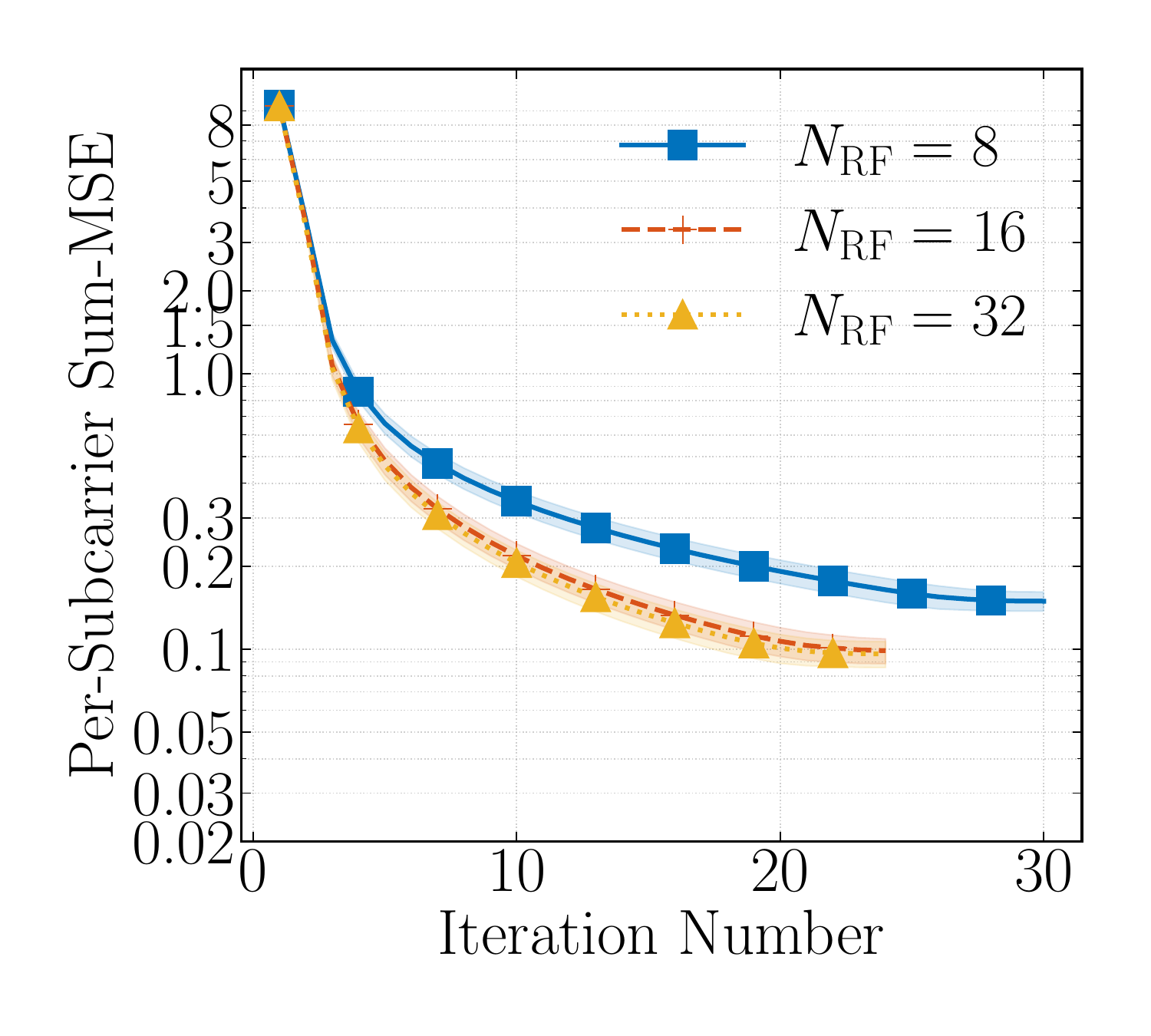}
	}\hfil
	\subfloat[\label{fig:rate_power_32}]{
		\includegraphics[width=0.235\textwidth]{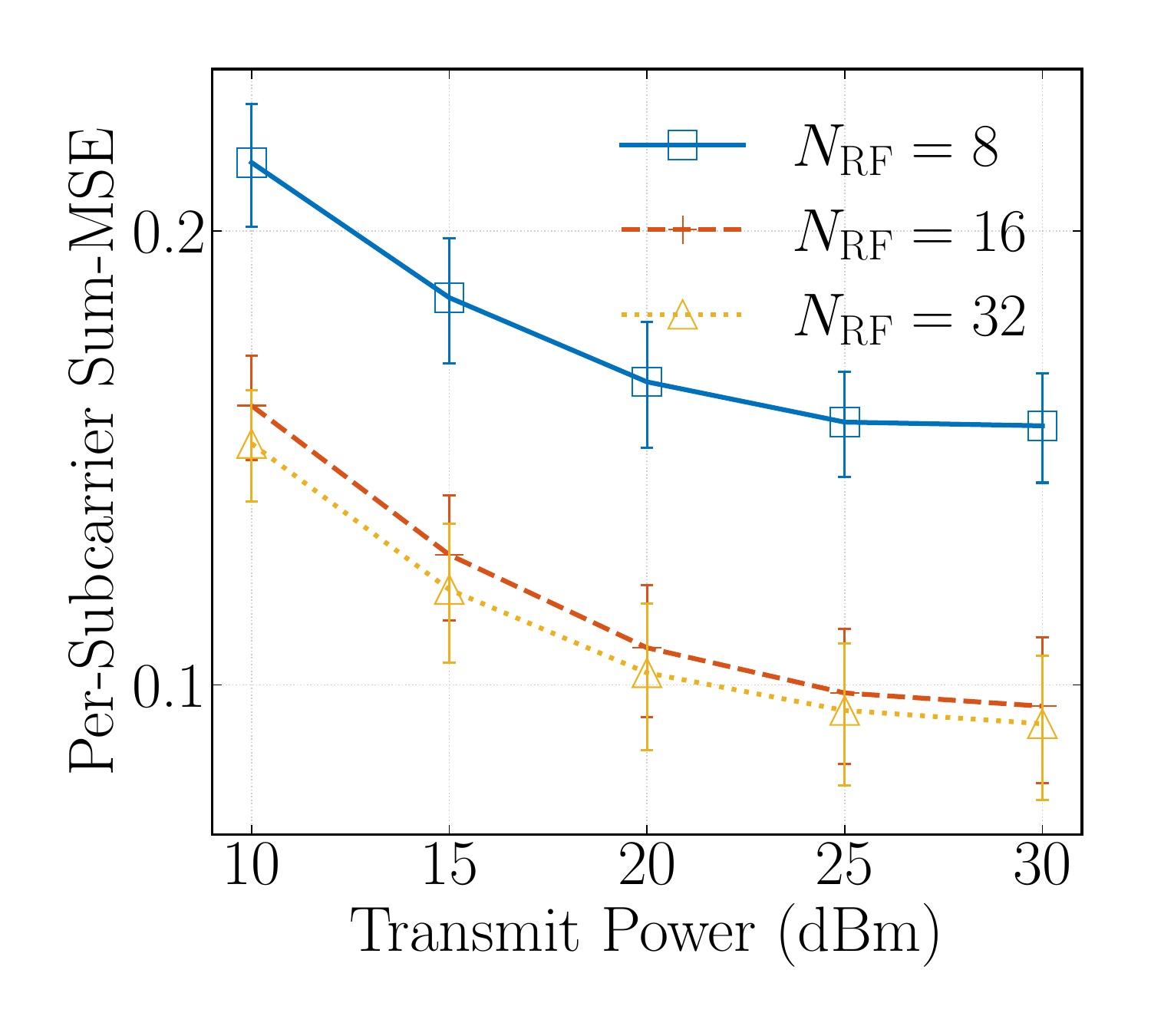}
	}\hfil
	\subfloat[\label{fig:rate_power_64_antenna1}]{
		\includegraphics[width=0.235\textwidth]{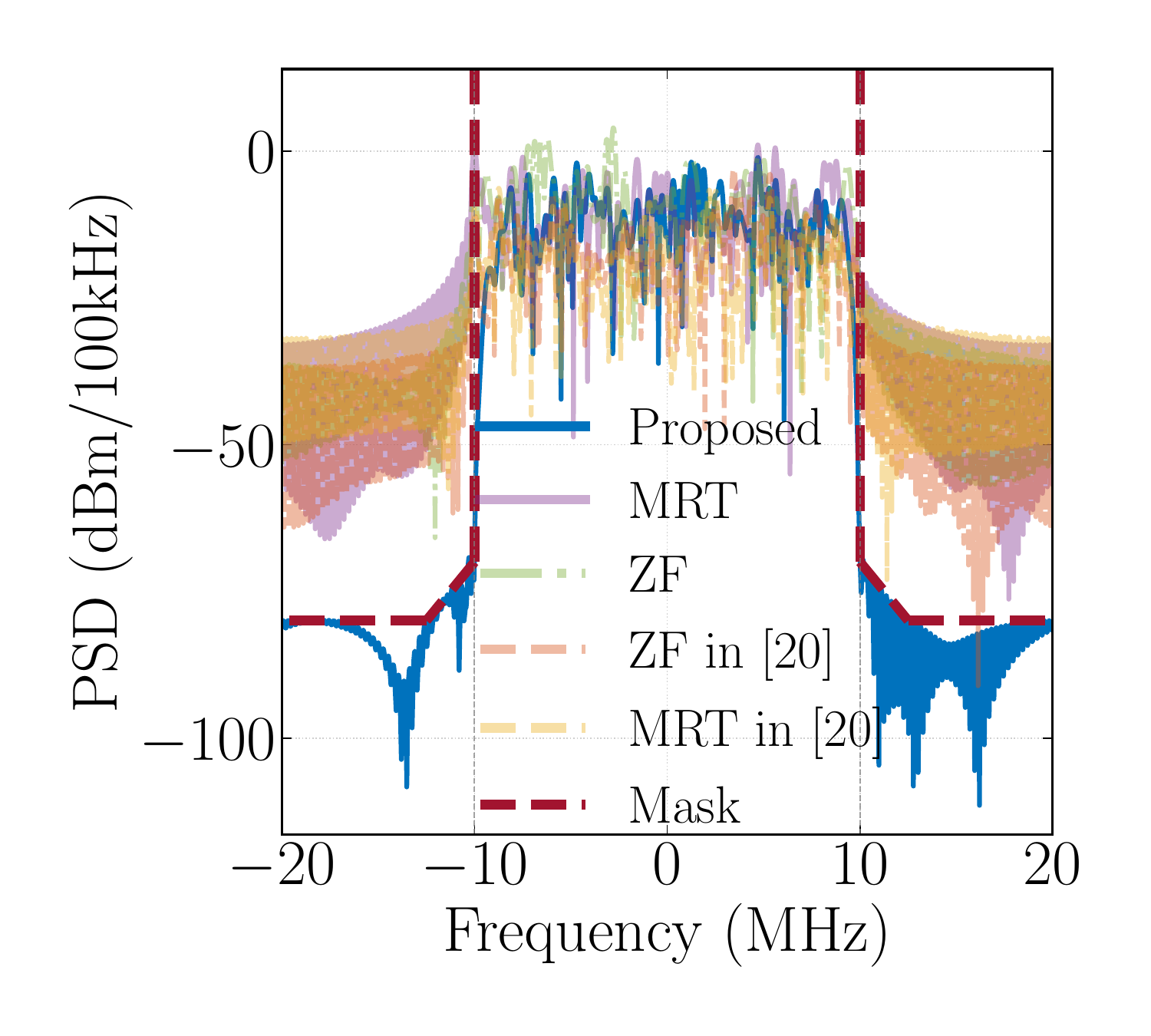}
	}\hfil
	\subfloat[\label{fig:MSE}]{
		\includegraphics[width=0.235\textwidth]{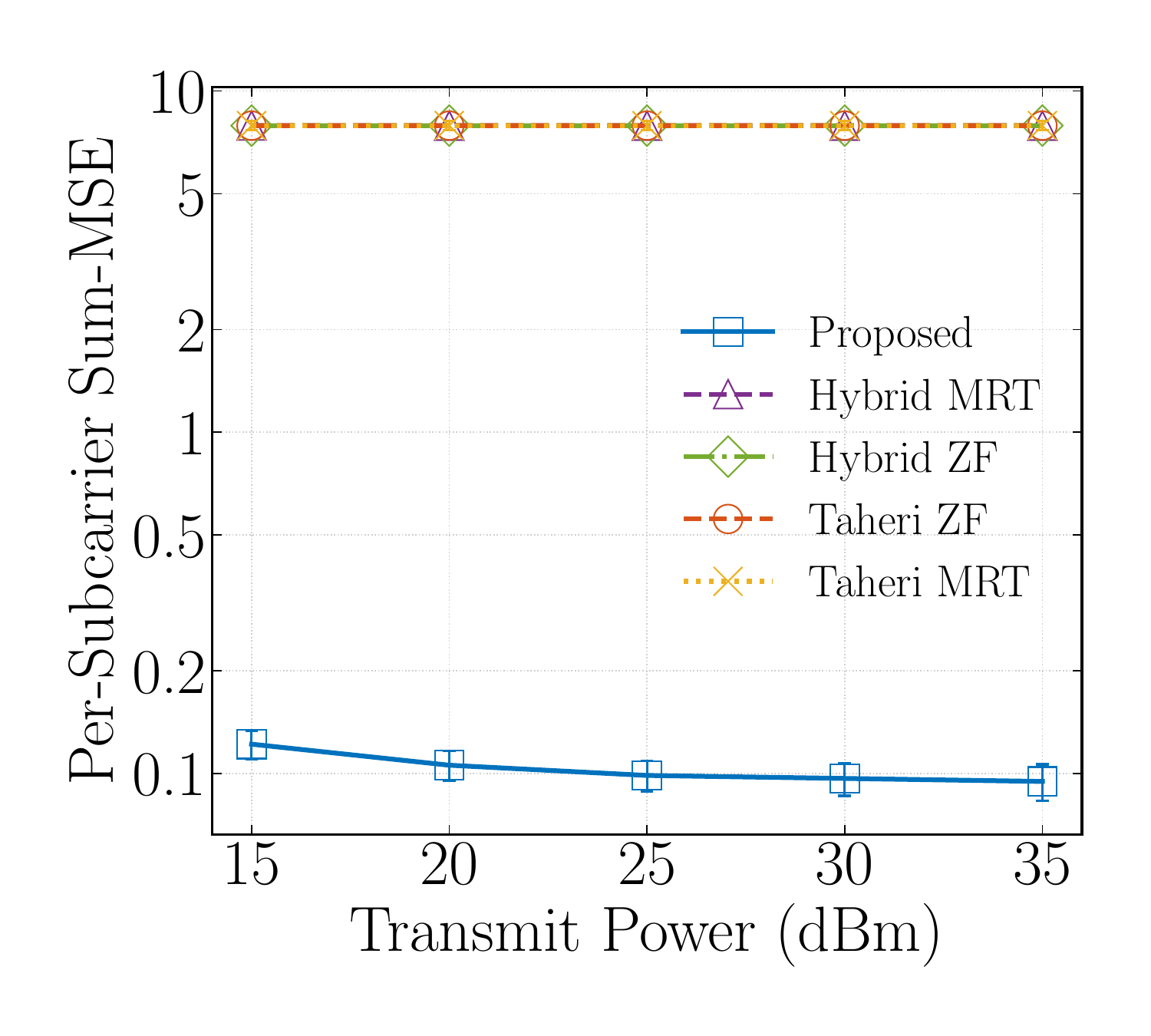}
	}
	\caption{Convergence and performance of the proposed scheme.
		(a)~Convergence of the proposed BCD approach.
		(b)~Average per-subcarrier sum-MSE with $S = 64$.
		(c)~Mask and PSD of transmitted OFDM symbol.
		(d)~Per-subcarrier sum-MSE.}
	\label{fig:summary_panels}
\end{figure*}

\begin{proposition}
	\label{prop:bcd_stationarity_hyb}
	Suppose the assumptions of Proposition~\ref{prop:admm_convergence_hyb} hold. Consider the cyclic outer BCD algorithm that, at each iteration~$\tau$, performs the exact MMSE update~\eqref{eq:U} for $\{\mathbf U_k^s\}$, computes the exact solution of~\eqref{opt:inner4block_hyb} for $(\mathbf Q,\mathbf X,\mathbf W,\mathbf T)$, and solves the RF-precoder and analog-combiner subproblems (Sections~\ref{sec:analogp}--\ref{sec:updateU}) by coordinate descent until stationarity. Let $\boldsymbol{\zeta}^\tau$ collect all primal variables after cycle~$\tau$, and assume $\lambda_{\min}\!\bigl(\mathbf U_{\mathrm{RF},k}^{\tau\,H}\mathbf U_{\mathrm{RF},k}^{\tau}\bigr)\ge\underline\lambda_k>0$ (where $\lambda_{\min}(\cdot)$ denotes the smallest eigenvalue of its matrix argument) for all~$k,\tau$. Then, \emph{(i)}~the regularized objective is monotonically nonincreasing and convergent, and $\{\boldsymbol{\zeta}^\tau\}$ is bounded and hence has accumulation points; \emph{(ii)}~let $\boldsymbol{\zeta}^{\tau,0}:=\boldsymbol{\zeta}^{\tau}$, let $\boldsymbol{\zeta}^{\tau,r}$ be the point after block $r$ in cycle $\tau$, and let $\boldsymbol{\zeta}^{\tau,R}=\boldsymbol{\zeta}^{\tau+1}$. If $\|\boldsymbol{\zeta}^{\tau,r}-\boldsymbol{\zeta}^{\tau,r-1}\|\to 0$ for all $r=1,\ldots,R$, then every accumulation point of $\{\boldsymbol{\zeta}^\tau\}$ is coordinatewise stationary for the regularized problem; \emph{(iii)}~if independent runs indexed by $\nu$ use $(\eta_w^\nu,\eta_t^\nu)\to(0,0)$ with $(\eta_w^\nu,\eta_t^\nu)>\mathbf 0$, and if $\boldsymbol{\zeta}^\nu$ denotes a coordinatewise stationary point returned by the $\nu^{\text{th}}$ run, then every accumulation point of $\{\boldsymbol{\zeta}^\nu\}$ is coordinatewise stationary for the unregularized problem, provided the same bounds $\underline\lambda_k$ hold uniformly for all runs.
\end{proposition}

\begin{proof}
	\emph{Part~(i).} Each block update is nonincreasing: $\{\mathbf U_k^s\}$ is minimized exactly by~\eqref{eq:U}, $(\mathbf Q,\mathbf X,\mathbf W,\mathbf T)$ by Proposition~\ref{prop:admm_convergence_hyb}, and the RF-precoder and analog-combiner sweeps by exact scalar minimization over unit-modulus sets. Since the objective is bounded below, it converges. Boundedness of $\{\boldsymbol{\zeta}^\tau\}$ follows because $\mathbf T^\tau\in\mathcal P$ is bounded, the equalities in~\eqref{opt:inner4block_hyb} propagate boundedness to $\mathbf W^\tau,\mathbf Q^\tau,\mathbf X^\tau$, the RF precoder and analog combiners lie in compact unit-modulus sets, and~\eqref{eq:U} yields $\|\mathbf U_k^{s,\tau}\|_F\le\|\mathbf U_{\mathrm{RF},k}^{\tau}\|_2\|\mathbf H_k^s\|_2\|\boldsymbol\omega_k^s\|_2\sqrt{P^s}/(\sigma_{\mathrm{noise},k}^{s\,2}\underline\lambda_k)$. Hence, $\{\boldsymbol{\zeta}^\tau\}$ is bounded and has accumulation points.
	
	\emph{Part~(ii).} Let $\bar{\boldsymbol{\zeta}}$ be any accumulation point, with $\boldsymbol{\zeta}^{\tau_j}\to\bar{\boldsymbol{\zeta}}$ along a subsequence. Since $\|\boldsymbol{\zeta}^{\tau,r}-\boldsymbol{\zeta}^{\tau,r-1}\|\to 0$ for all $r$, we also have $\boldsymbol{\zeta}^{\tau_j,r}\to\bar{\boldsymbol{\zeta}}$ for every block $r$. For the unconstrained/convex blocks $\{\mathbf U_k^s\}$ and $(\mathbf Q,\mathbf X,\mathbf W,\mathbf T)$, the exact first-order/KKT conditions hold at $\boldsymbol{\zeta}^{\tau_j,r}$ and are preserved in the limit by continuity of the smooth terms and outer semicontinuity of the normal-cone mappings for $\mathcal M$, $\mathcal C$, and $\mathcal P$. For the RF-precoder and analog-combiner blocks, coordinatewise stationarity means each scalar feasible entry minimizes its coordinate objective over the unit circle with all other variables fixed; since these objectives are continuous and the unit circle is compact, the scalar optimality inequalities are preserved in the limit. Hence $\bar{\boldsymbol{\zeta}}$ is coordinatewise stationary for the regularized problem.
	
	\emph{Part~(iii).} The bounds from Part~(i) ensure $\{\boldsymbol{\zeta}^\nu\}$ is bounded; let $\boldsymbol{\zeta}^{\nu'}\to\bar{\boldsymbol{\zeta}}$ along a convergent subsequence. Each $\boldsymbol{\zeta}^{\nu'}$ satisfies the blockwise stationarity conditions of the regularized problem. The only extra terms relative to the unregularized problem are $\eta_w^{\nu'}\mathbf W^{\nu'}$ and $\eta_t^{\nu'}\mathbf T^{\nu'}$, which vanish since $\{\mathbf W^{\nu'}\},\{\mathbf T^{\nu'}\}$ are bounded and $(\eta_w^{\nu'},\eta_t^{\nu'})\to(0,0)$. Passing to the limit by the same continuity and compactness arguments as in Part~(ii) (outer semicontinuity of normal cones for the convex blocks, and preservation of scalar optimality inequalities on the unit circle for the RF-precoder and analog-combiner blocks) shows that $\bar{\boldsymbol{\zeta}}$ is coordinatewise stationary for the unregularized problem. As the convergent subsequence was arbitrary, the same argument applies to every accumulation point of $\{\boldsymbol{\zeta}^\nu\}$.
\end{proof}
	
		\subsection{Complexity Analysis}
	\label{sec:comp}
In this section, we quantify the computational complexity of the sequential optimization. For the RF precoder, updating a single PS entry has per-iteration cost $\mathcal{O}(N_t)$, and updating the full RF precoder once (i.e., sweeping all PS entries) has cost $\mathcal{O}(N_t^2)$. For the analog combiner update, computing $\beta_{1, a m}$ costs $\mathcal{O}(S)$, while computing $\beta_{2, a m}$ costs $\mathcal{O}(S n_k N_r)$; consequently, the per-variable update cost for the analog combiner is $\mathcal{O}(S n_k N_r)$. In the ADMM routine, all subproblems admit closed-form updates except for $\mathbf t$. Updating $\mathbf t$ within each bisection step requires one $\mathcal{O}(N_{\text{RF}}^3)$ inversion of an $N_{\text{RF}}\times N_{\text{RF}}$ matrix, followed by a matrix multiplication with cost $\mathcal{O}(N_{\text{RF}} n_k)$. The per-antenna update costs for $\mathbf q^a$, $\mathbf w^a$, and $\mathbf x^a$ are $\mathcal{O}(GS)$, $\mathcal{O}(S^2)$ after a one-time $\mathcal{O}(S^3)$ factorization of the common matrix $\mathbf M_w$, and $\mathcal{O}(\ell S^2)$, respectively. When $G\ll S$, the update of $\mathbf w^a$ can alternatively be implemented exactly via the matrix inversion lemma, replacing the inversion of an $S\times S$ matrix by that of a $G\times G$ matrix. Finally, updating $\mathbf V_k^s$ for each user and subcarrier has cost $\mathcal{O}(N_{\text{RF}} n_k)$.
	
		\section{Simulation Results}
	\label{sec:sim}

\begin{figure*}[ht!]
	\centering
	\hspace{.5cm}\includegraphics[width=1.04\textwidth]{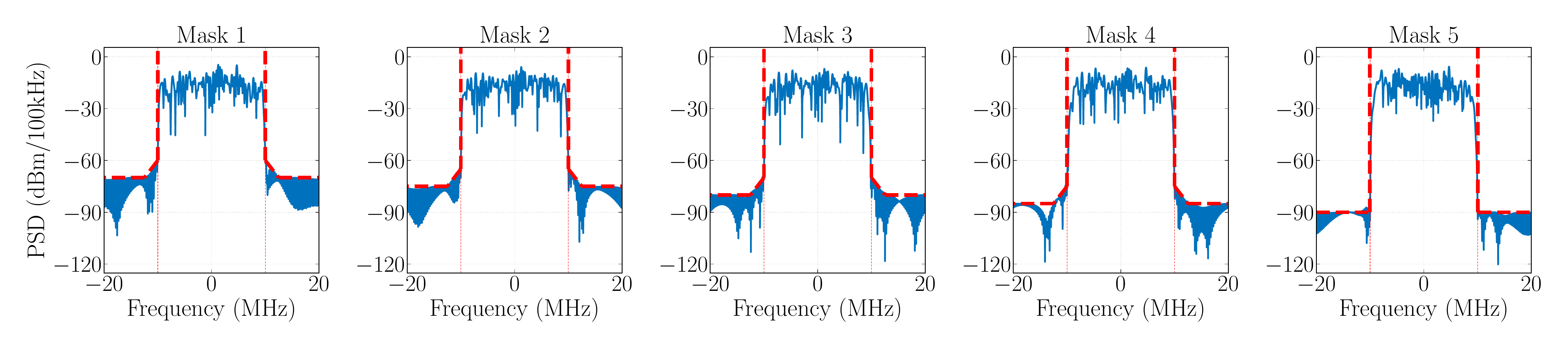}
	\caption{PSD of the transmitted 20-MHz-bandwidth OFDM symbol at antenna 1 for 5 different masks.}
	\label{fig:psd}
\end{figure*}

\begin{figure*}[t!]
	\centering
	\subfloat[\label{fig:emission_rate}]{
		\includegraphics[width=0.315\textwidth]{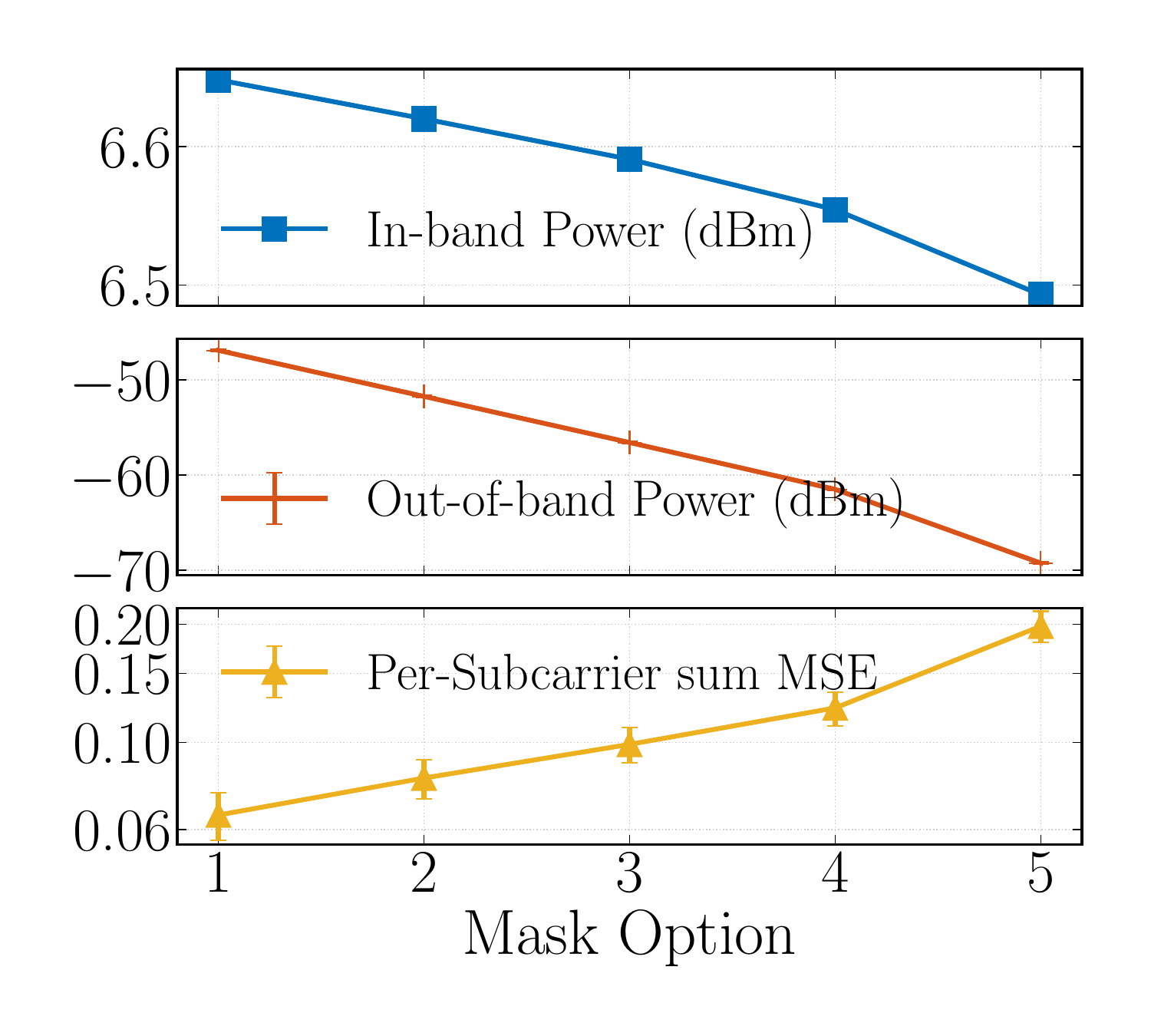}
	}\hfil
	\subfloat[\label{fig:ser}]{
		\includegraphics[width=0.315\textwidth]{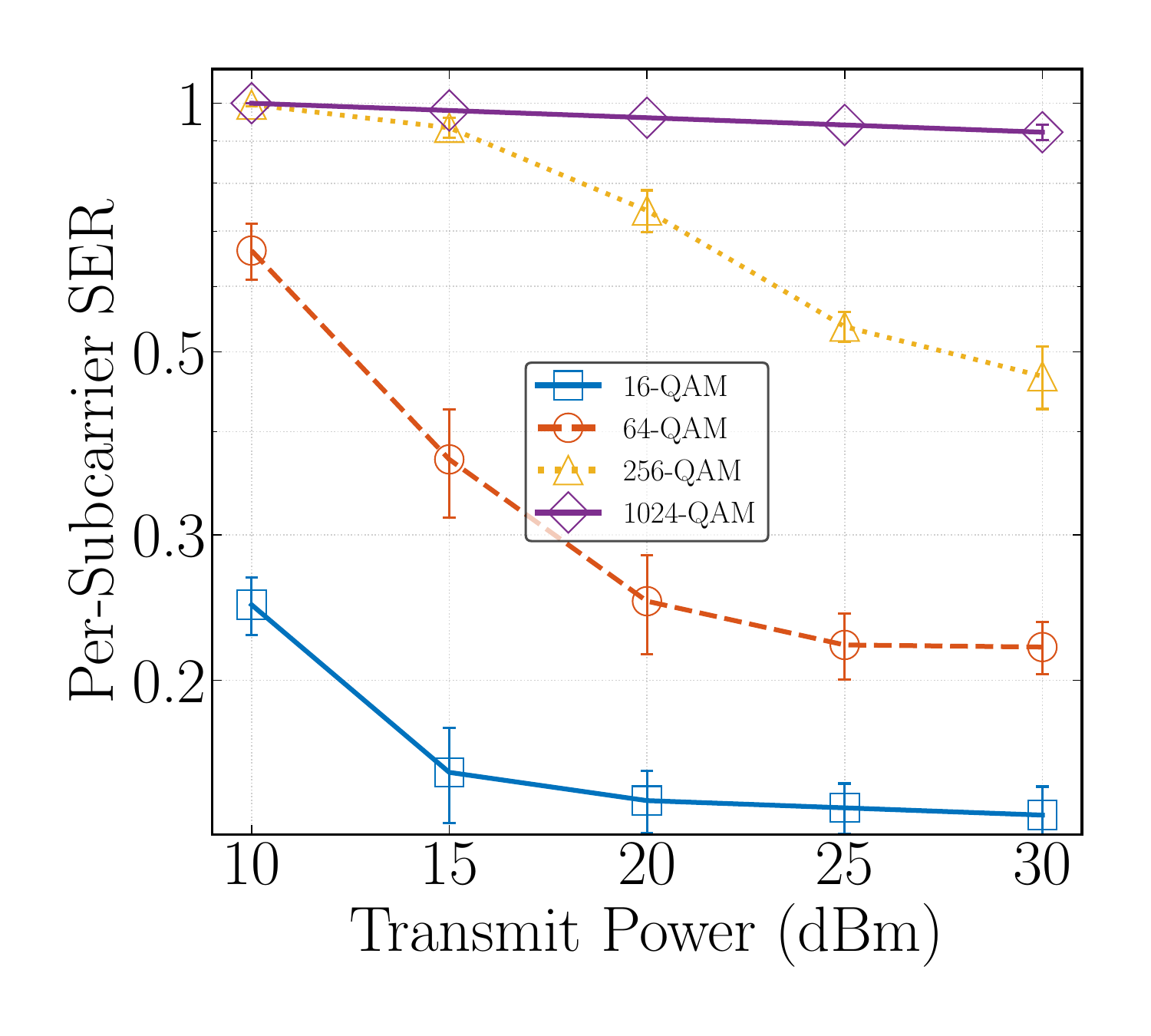}
	}\hfil
	\subfloat[\label{fig:evm}]{
		\includegraphics[width=0.315\textwidth]{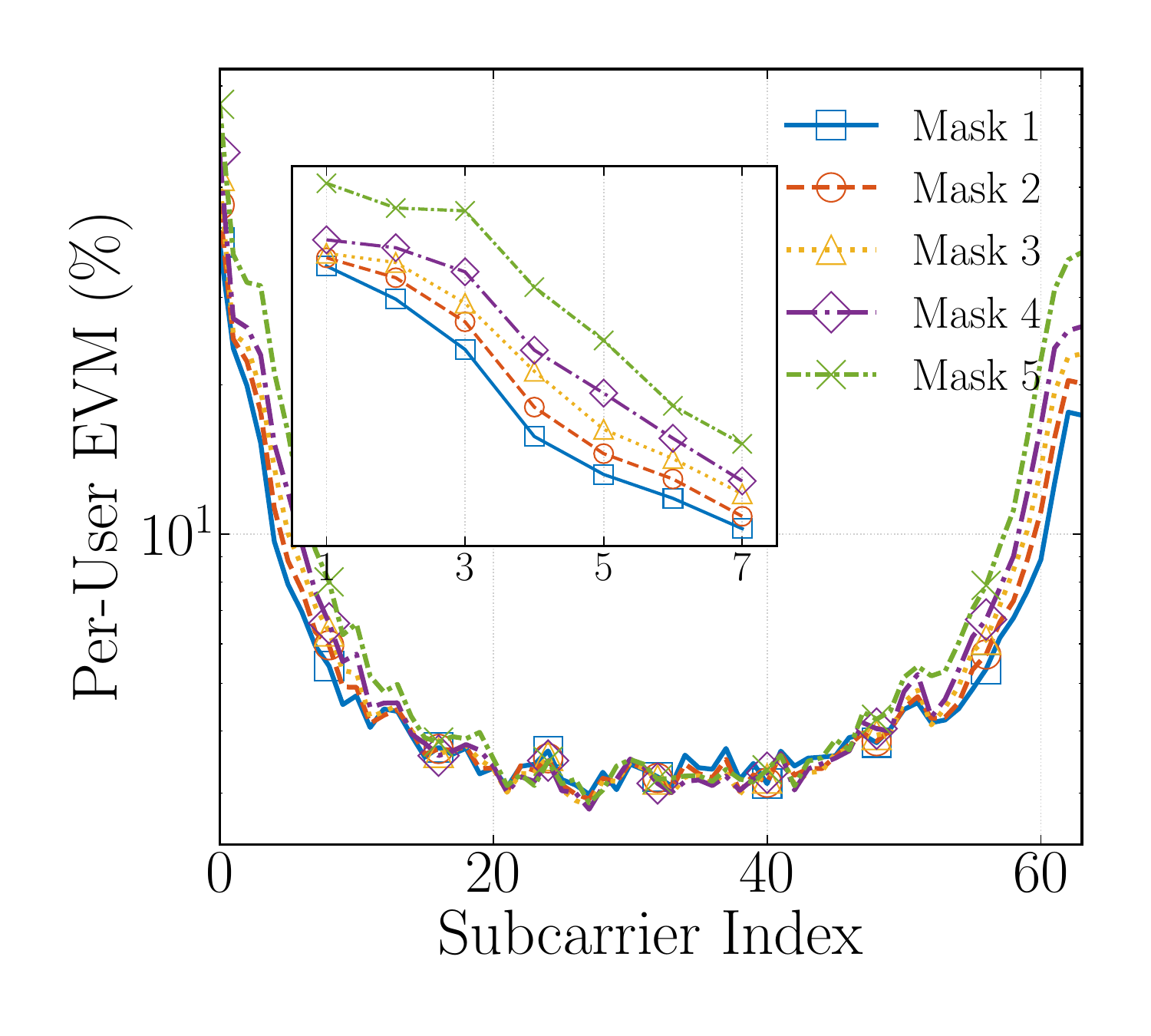}
	}
	\caption{
		(a)~In-band/out-of-band emitted power and the average per-subcarrier sum-MSE for 5 different masks.
		(b)~Per-subcarrier SER of the proposed method.
		(c)~Average per-user EVM versus subcarrier for different masks.}
\end{figure*}

Simulations are conducted to evaluate the proposed schemes in an MU downlink MIMO--OFDM system. The base station employs
$N_t=32$ antennas and serves $K=4$ users, where each user is equipped with $N_r=4$ antennas, $2$ RF chains, and $n_k=2$. The total bandwidth is
$20$~MHz with center frequency $f^c=28$~GHz and $S=64$, and all users occupy all subcarriers. The users are randomly distributed in a circle of radius \(4\) meters, located
\(100\) meters from the transmitter. The symbols are independently drawn from a normalized 64-QAM constellation, unless otherwise stated, and $B=30$.

Uniform linear arrays (ULAs) are assumed at both the transmitter and receivers, with responses
$\mathbf{a}_t(\theta)=\big[1,e^{\jmath\psi},e^{\jmath2\psi},\ldots,e^{\jmath(N_t-1)\psi}\big]^T$ and
$\mathbf{a}_r(\phi)=\big[1,e^{\jmath\psi'},e^{\jmath2\psi'},\ldots,e^{\jmath(N_r-1)\psi'}\big]^T$, where
$\psi=\frac{2\pi d_t\sin(\theta)}{\lambda}$, $\psi'=\frac{2\pi d_r\sin(\phi)}{\lambda}$,
and $\theta$, $\phi$ are measured from each array's broadside.
We set $d_t=d_r=\lambda/2$; at $28$~GHz, $\lambda\approx 10.7$~mm.
Each user's receive array is oriented with its broadside facing the transmitter,
and the LOS angle of departure $\theta_k^{\text{AoD}}$ is determined by
user~$k$'s geometric position relative to the transmitter array broadside.

A frequency-selective Rician MIMO--OFDM channel with $T$ taps is considered. For user $k$ on subcarrier $s$, the channel matrix $\mathbf{H}_k^s\in\mathbb{C}^{N_r\times N_t}$ is generated as~\cite{bjornson2024introduction} $\mathbf{H}_k^s=\sqrt{\frac{\kappa}{\kappa+1}}\sqrt{g_k}\,\mathbf{a}_r(\phi_k^{\text{AoA}})\mathbf{a}_t^H(\theta_k^{\text{AoD}})+\sum_{l=1}^{T-1}\sqrt{\frac{1}{\kappa+1}}\sqrt{g_{k,l}}\,\mathbf{a}_r(\phi_{k,l}^{\text{AoA}})\mathbf{a}_t^H(\theta_{k,l}^{\text{AoD}})\,h_{k,l}\,e^{-\jmath 2\pi l s/S}$. Here, $h_{k,l}\sim\mathcal{CN}(0,1)$. The Rician factor is $\kappa_{\text{dB}}=10$~dB.

The large-scale fading follows the 3GPP model in \cite{3GPP_TR_36_814_2017}. The LOS path loss is
$PL_{\text{dB}}^{\text{LOS}}=22\log_{10}(d_k)+28+20\log_{10}(f^c)+\xi_{\text{LOS}}$, where $d_k$ is the user distance (m),
$f^c$ is the carrier frequency (GHz), and $\xi_{\text{LOS}}\sim\mathcal{N}(0,\sigma_{\text{LOS}}^2)$ with $\sigma_{\text{LOS}}=5.8$~dB; hence,
$g_k=10^{-PL_{\text{dB}}^{\text{LOS}}/10}$. For NLOS components,
$PL_{\text{dB}}^{\text{NLOS}}=22\log_{10}(d_k)+28+20\log_{10}(f^c)+\xi_{\text{NLOS}}$, where
$\xi_{\text{NLOS}}\sim\mathcal{N}(0,\sigma_{\text{NLOS}}^2)$ with $\sigma_{\text{NLOS}}=8.7$~dB, and
$g_{k,l}=10^{-PL_{\text{dB}}^{\text{NLOS}}/10}$. For NLOS paths, departure and arrival angles are generated around the LOS directions as
$\theta_{k,l}^{\text{AoD}}=\theta_k^{\text{AoD}}+\Delta\theta_{k,l}$ and
$\phi_{k,l}^{\text{AoA}}=\phi_k^{\text{AoA}}+\Delta\phi_{k,l}$, where
$\Delta\theta_{k,l},\Delta\phi_{k,l}\sim\mathcal{N}(0,\sigma_\theta^2)$. Additive noise at user $k$ on subcarrier $s$ is modeled as
$\mathbf{n}_k^s\sim\mathcal{CN}\!\big(\mathbf{0},\sigma_{\text{noise},k}^{s\,2}\mathbf{I}_{N_r}\big)$.
The noise power spectral density is $-174$~dBm/Hz and the receiver noise figure is $8$~dB.
	
We set the maximum admissible amplitude of the unclipped waveform to $\chi=3$. Unless stated otherwise, we use a mask that is inactive for $|f|<10.01$~MHz, tightens linearly from $-70$ to $-80$~dBm/$100$~kHz over $|f|\in[10.01,12.5]$~MHz, and remains flat at $-80$~dBm/$100$~kHz for $|f|\ge 12.5$~MHz. We enforce the mask (and the notches over $[-18,-10.01]$~MHz and $[10.01,18]$~MHz) using $90$ uniformly spaced frequency samples per side to form $\mathbf A_n$. For $64$ subcarriers and a per-subcarrier power
budget of $25$~dBm, Fig.~\ref{fig:conv_bcd} depicts the per-subcarrier convergence of the proposed BCD algorithm for different numbers
of RF chains. The sum-MSE decreases monotonically with the iteration index. For the $64$-subcarrier configuration, the user-averaged per-subcarrier sum-MSE is shown in
Fig.~\ref{fig:rate_power_32} as the per-subcarrier transmit power budget is swept from $10$~dBm to $30$~dBm in $5$~dB steps; the average
per-subcarrier sum-MSE decreases with increasing transmit power and is further reduced when more RF chains are employed. For the case $N_{\text{RF}}=16$ and $P^s=30$ dBm, Fig.~\ref{fig:rate_power_64_antenna1} shows the PSD of the transmitted OFDM symbol from antenna~1 together with the corresponding spectral mask. By enforcing the PSD constraints at 90 frequency samples on each side of the band, the resulting spectrum remains below the mask over the entire OOB region. For comparison, conventional ZF and MRT are also included, along with ZF and MRT combined with frequency notching based on the same $\mathbf{A}_n$ as in \cite{taheri2020joint}. It can be seen that the proposed method achieves the lowest sidelobes. These approaches are compared in minimizing the sum-MSE for different $P^s$ values in Fig. \ref{fig:MSE}. The proposed method not only yields lower OOB emission, but also achieves significantly lower sum-MSE values. This is due the fact that the proposed method deployes an optmized combiner at each user, by which the received signal is strengthened and decoded.

Fig.~\ref{fig:psd} compares five OOB mask profiles. For Masks~1--4, the mask is inactive for $|f|<10.01$~MHz; over
$|f|\in[10.01,12.5]$~MHz it tightens linearly from $-60/-65/-70/-75$~dBm/$100$~kHz to $-70/-75/-80/-85$~dBm/$100$~kHz; and it remains
constant at $-70/-75/-80/-85$~dBm/$100$~kHz for $|f|\ge 12.5$~MHz. Mask~5 starts at $|f|=10.01$~MHz and enforces a flat limit of
$-90$~dBm/$100$~kHz for $|f|\ge 10.01$~MHz.

The in-band emissions, OOB emissions, and the average  per-subcarrier sum-MSE are depicted in Fig.~\ref{fig:emission_rate} for
$P^s=25$~dBm and $N_{\mathrm{RF}}=16$. Under Mask~1, the optimized hybrid precoder achieves $6.65$~dBm in-band emissions and $-46.93$~dBm OOB
emissions, with an average  per-subcarrier sum MSE of $6.54\times 10^{-2}$. Tightening the mask progressively reduces OOB emissions from
$-46.93$~dBm (Mask~1) to $-51.76$~dBm (Mask~2), $-56.59$~dBm (Mask~3), $-61.51$~dBm (Mask~4), and $-69.25$~dBm (Mask~5), while the in-band
emissions remain nearly unchanged at $6.49$--$6.65$~dBm. This improved OOB suppression comes at the cost of increased MSE, as the average
 per-subcarrier sum MSE increases from $6.54\times 10^{-2}$ (Mask~1) to $8.12\times 10^{-2}$ (Mask~2), $9.89\times 10^{-2}$ (Mask~3),
$1.225\times 10^{-1}$ (Mask~4), and $1.977\times 10^{-1}$ (Mask~5).

The per-subcarrier symbol-error rate (SER) for different QAM orders is shown in Fig.~\ref{fig:ser}. As expected, the SER decreases as the modulation order is reduced.

The average per-user error vector magnitude (EVM) across subcarriers is shown in Fig.~\ref{fig:evm}. The EVM on the central subcarriers is lower than that on
the edge subcarriers, since the imposed masks require stronger suppression near the band edges, which reduces precoding gains and increases the
sum-MSE to satisfy the OOB constraints.

\begin{figure}
	\centering
	\includegraphics[width=.37\textwidth]{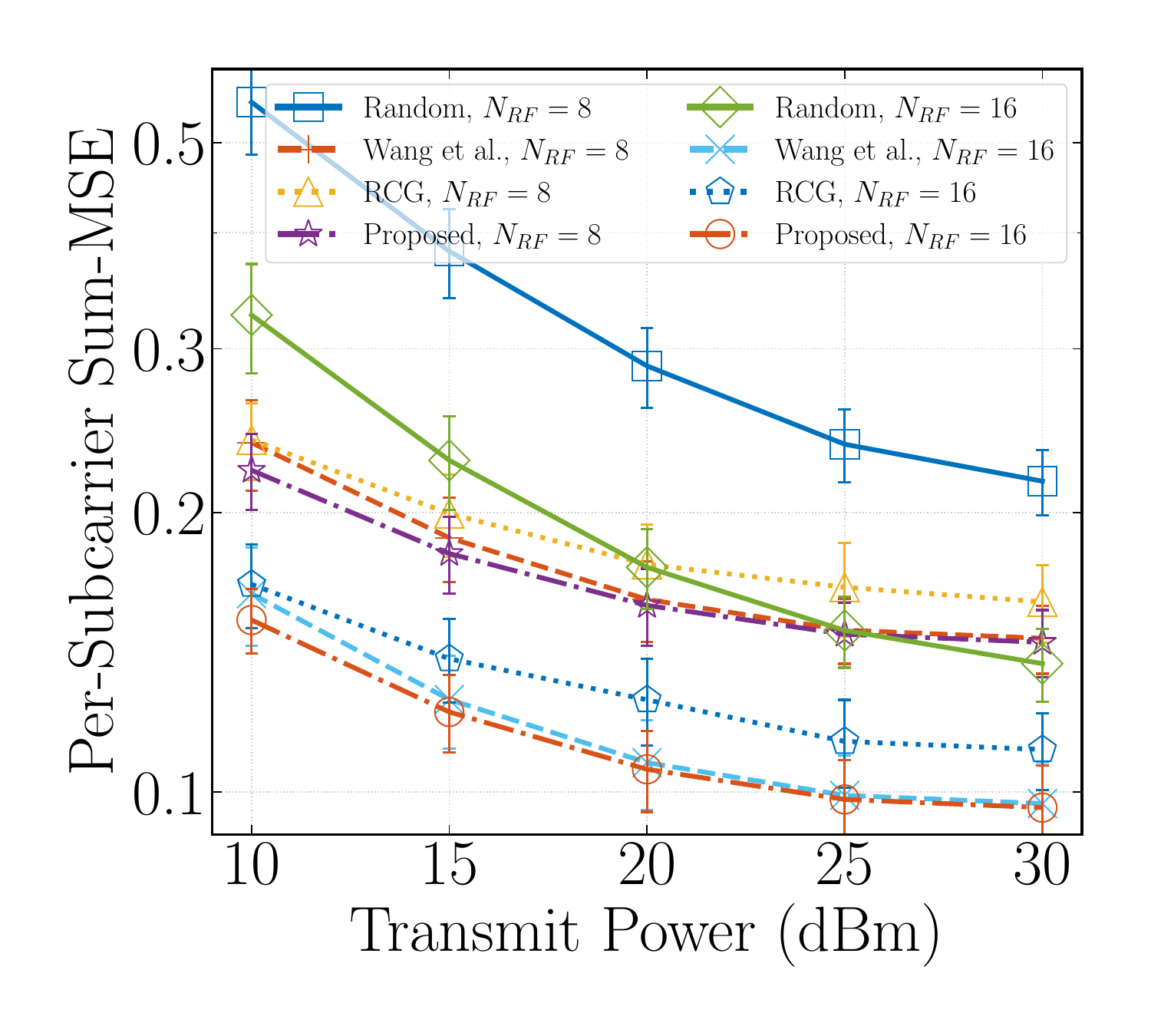}
	\caption{Average per-subcarrier sum-MSE of different hybrid precoding methods.}\label{fig:analog_rate}
\end{figure}
	
		\begin{figure}
		\centering
		\includegraphics[width=.37\textwidth]{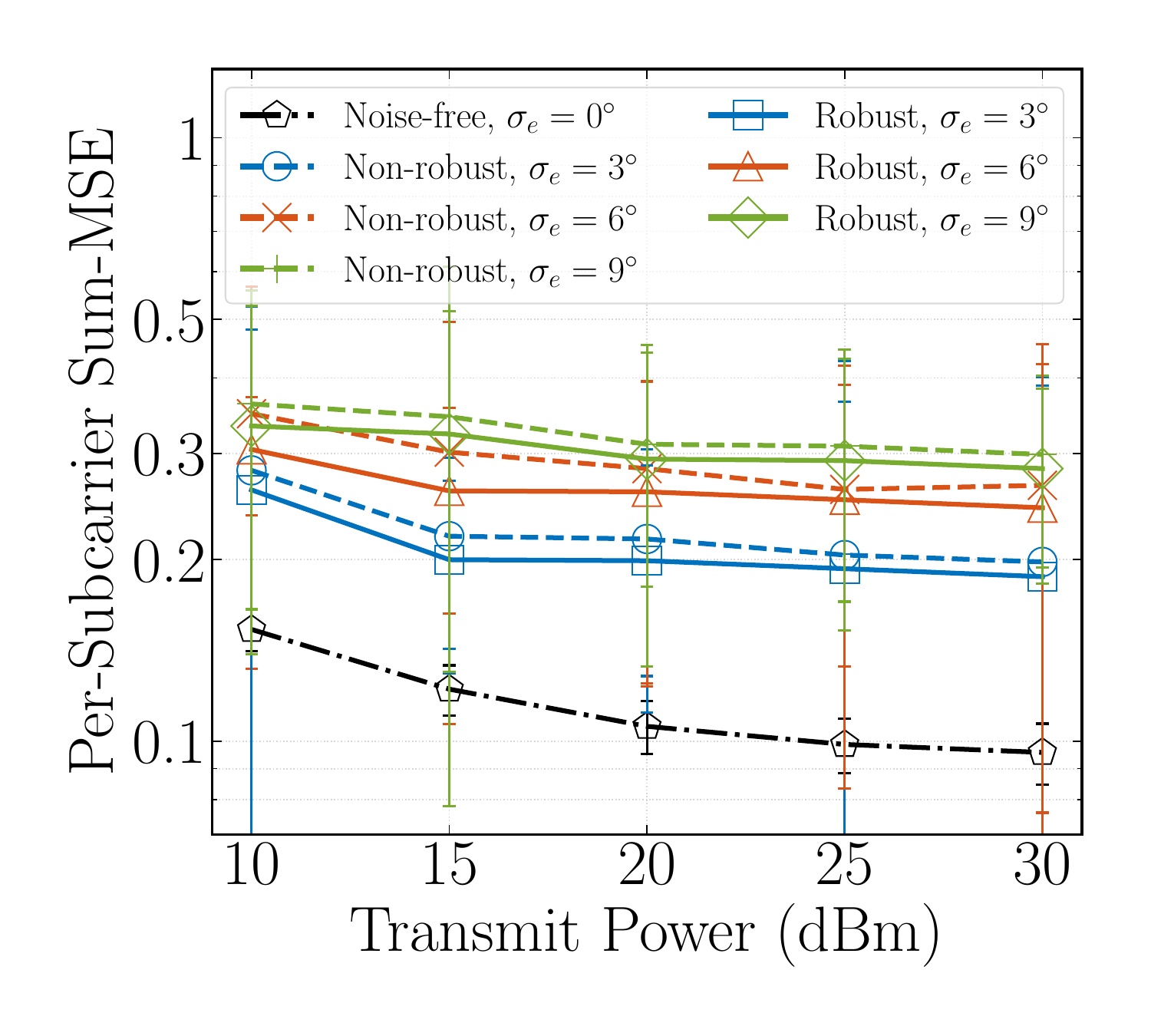}
		\caption{The proposed BCD methods with different PS noise levels when $N_{\text{RF}}=16$.}\label{fig:robust1}
	\end{figure}
We compare the proposed RF precoding--combining design with the coordinate-descent approach of Wang \textit{et al.}~\cite{wang2018hybrid}, which
iteratively reduces the objective by alternately updating individual PSs at the transmitter or for a given user through a numerical line search,
while keeping all remaining PS values fixed. We also include the RCG-based method considered in \cite{7397861,9467491}, and a baseline with
randomly assigned PS values for both the precoder and the combiners. Fig.~\ref{fig:analog_rate} depicts the average per-subcarrier sum-MSE
for different numbers of RF chains, where the proposed method consistently achieves the lowest sum-MSE among all approaches.

Fig.~\ref{fig:robust1} compares the proposed robust design with the non-robust (with PS update rules not accounting for PS errors) schemes for
$N_{\mathrm{RF}}=16$. We evaluate PS impairments by considering phase-noise levels $\sigma_e\in\{3^\circ,6^\circ,9^\circ\}$. As shown
in Fig.~\ref{fig:robust1}, phase noise degrades the average per-subcarrier sum-MSE for all methods. However, when the proposed robust PS
update rules are employed, the resulting algorithm consistently achieves lower sum-MSE than its non-robust counterpart in the presence of PS impairments.

	\section{Concluding Remarks}
	\label{sec:con}
	
In this work, we studied hybrid RF--digital precoder and combiner design for MU-MIMO--OFDM while accounting for PS impairments and explicitly enforcing amplitude clipping and OOB spectral mask constraints per OFDM symbol realization. We developed an MMSE-based BCD algorithm that minimizes the sum-MSE across all users and subcarriers. The resulting nonconvex problem is addressed block by block: the digital precoder update is carried out via an ADMM routine whose substeps admit closed-form solutions or simple bisection searches, while the RF precoder and the users' analog combiners are updated using coordinate-descent rules with closed-form phase updates that satisfy the unit-modulus constraints, including robust variants under PS errors. All blocks are updated in a Gauss--Seidel fashion until the stopping criterion is met. Simulation results demonstrate consistent performance gains over standard baseline methods.

	\ifCLASSOPTIONcaptionsoff
	\newpage
	\fi

	\bibliographystyle{IEEEbib}
	\bibliography{ref_SDRA}

@article{tao2018convergence,
		title={Convergence analysis of the direct extension of {ADMM} for multiple-block separable convex minimization},
		author={M. Tao and X. Yuan},
		journal={Adv. Comput. Math.},
		volume={44},
		number={3},
		pages={773--813},
		year={Jun. 2018},
		publisher={Springer}
	}

@misc{fcc_uwb_15_517,
	author       = {{FCC}},
	title        = {{47 C.F.R. \S 15.517}},
	howpublished = {\url{https://www.ecfr.gov/current/title-47/part-15/section-15.517}},
	note         = {eCFR, accessed Mar. 29, 2026}
}

@misc{fcc_cbrs_96_41,
	author       = {{FCC}},
	title        = {{47 C.F.R. \S 96.41}},
	howpublished = {\url{https://www.ecfr.gov/current/title-47/part-96/section-96.41}},
	note         = {eCFR, accessed Mar. 29, 2026}
}

@article{bjornson2024introduction,
		title={Introduction to multiple antenna communications and reconfigurable surfaces},
		author={E. Bj{\"o}rnson and {\"O}. Demir},
		year={2024},
		publisher={Now Publishers, Inc.}
	}

@ARTICLE{7956180,
		author={R. Garg and A. S. Natarajan},
		journal={IEEE Trans. Microw. Theory Tech.}, 
		title={A 28-{GHz} Low-Power Phased-Array Receiver Front-End With 360$^\circ$ {RTPS} Phase Shift Range}, 
		year={Jun. 2017},
		volume={65},
		number={11},
		pages={4703-4714},
		doi={10.1109/TMTT.2017.2707414}}

@ARTICLE{7948784,
			author={C.W. Byeon and C.S. Park},
			journal={IEEE Microw. Wireless Compon. Lett}, 
			title={A Low-Loss Compact 60-{GHz} Phase Shifter in 65-nm {CMOS}}, 
			year={Jun. 2017},
			volume={27},
			number={7},
			pages={663-665},
			doi={10.1109/LMWC.2017.2711569}}

@ARTICLE{9767793,
		author={J. Li and Z. Wang and Y. Zhang and P. Zhu and D. Wang and X. You},
		journal={IEEE Syst. J.}, 
		title={Robust Hybrid Beamforming for Outage-Constrained Multigroup Multicast {mmWave} Transmission With Phase Shifter Impairments}, 
		year={Mar. 2023},
		volume={17},
		number={1},
		pages={869-880},
		doi={10.1109/JSYST.2022.3168022}}

@techreport{3GPP_TR_36_814_2017,
	author    = {{3GPP TR 36.814}},
	title     = {Further advancements for {E-UTRA} physical layer aspects},
	institution = {3rd Generation Partnership Project (3GPP)},
	year      = {2017},
	month     = {Mar.},
	note      = {Release 9, V9.2.0}
}

@ARTICLE{10778226,
		author={M. Alouzi and H. Yanikomeroglu and G. Karabulut Kurt},
		journal={IEEE Trans. Wireless Commun.}, 
		title={Adaptive Phase Shifters for Hybrid Beamforming in {mmWave} Systems}, 
		year={Feb. 2025},
		volume={24},
		number={2},
		pages={1104-1116},
		doi={10.1109/TWC.2024.3505202}}

@ARTICLE{10972245,
		author={H. Vaezy and S.D. Blostein},
		journal={IEEE Trans. Signal Process.}, 
		title={Joint User Selection and Hybrid Precoder Design for Massive {MIMO} Systems}, 
		year={Apr. 2025},
		volume={73},
		number={},
		pages={1808-1822},
		doi={10.1109/TSP.2025.3562858}}

@ARTICLE{10839437,
		title={Hybrid Precoding for {mmWave} Massive {MIMO} With Finite Blocklength},
		author={X. Zhang and L. Xiang and J. Wang and P. Zhu and D.W.K. Ng and X. Gao},
		journal={IEEE Trans. Commun.}, 
		year={Aug. 2025},
		volume={73},
		number={8},
		pages={6379-6395},
		doi={10.1109/TCOMM.2025.3529244}}

@ARTICLE{8048606,
	author={D. Nguyen and L. Le Bao and T.  Le-Ngoc and R.W. Heath},
	journal={IEEE Access}, 
	title={Hybrid {MMSE} Precoding and Combining Designs for {mmWave} Multiuser Systems}, 
	year={Sep. 2017},
	volume={5},
	number={},
	pages={19167-19181},
	doi={10.1109/ACCESS.2017.2754979}}

@ARTICLE{9151131,
	author={J. Galaviz-Aguilar V. Alejandro and T. Cesar and E. Tlelo-Cuautle},
	journal={IEEE Access}, 
	title={{RF-PA} Modeling of {PAPR}: A Precomputed Approach to Reinforce Spectral Efficiency}, 
	year={Jul. 2020},
	volume={8},
	number={},
	pages={138217-138235},
	doi={10.1109/ACCESS.2020.3012610}}

@ARTICLE{10048700,
		author={S. Liu and Y. Wang and Z. Lian and Y. Su and Z. Xie},
		journal={IEEE Trans. Broadcast.}, 
		title={Joint Suppression of {PAPR} and {OOB} Radiation for {OFDM} Systems}, 
		year={Feb. 2023},
		volume={69},
		number={2},
		pages={528-537},
		doi={10.1109/TBC.2023.3243410}}

@ARTICLE{9069262,
	author={W. Wang and H. Yin and X. Chen and W. Wang},
	journal={IEEE Access}, 
	title={Robust and Low-Overhead Hybrid Beamforming Design With Imperfect Phase Shifters in Multi-User Millimeter Wave Systems}, 
	year={May 2020},
	volume={8},
	number={},
	pages={74002-74014},
	doi={10.1109/ACCESS.2020.2988267}}

@ARTICLE{9467491,
	author={X. Zhao and T. Lin and Y. Zhu and J. Zhang},
	journal={IEEE Trans. Wireless Commun.}, 
	title={Partially-Connected Hybrid Beamforming for Spectral Efficiency Maximization via a Weighted {MMSE} Equivalence}, 
	year={Jul. 2021},
	volume={20},
	number={12},
	pages={8218-8232},
	doi={10.1109/TWC.2021.3091524}}

@ARTICLE{1468466,
		author={M. Joham and W. Utschick and J.A. Nossek},
		journal={IEEE Trans. Signal Process.}, 
		title={Linear transmit processing in {MIMO} communications systems}, 
		year={Aug. 2005},
		volume={53},
		number={8},
		pages={2700-2712},
		doi={10.1109/TSP.2005.850331}}

@article{van2009sculpting,
		title={Sculpting the multicarrier spectrum: a novel projection precoder},
		author={J. Van De Beek},
		journal={IEEE Commun. Lett.},
		volume={13},
		number={12},
		pages={881--883},
		year={Dec. 2009},
		publisher={IEEE}
	}

@ARTICLE{9390405,
		author={S. Kant and M. Bengtsson and B. Göransson and G. Fodor and C. Fischione},
		journal={IEEE Trans. Wireless Commun.}, 
		title={Efficient Optimization for Large-Scale {MIMO}-{OFDM} Spectral Precoding}, 
		year={Sep. 2021},
		volume={20},
		number={9},
		pages={5496-5513},
		doi={10.1109/TWC.2021.3068207}}

@ARTICLE{8653302,
	author={D. Mishra and H. Johansson},
	journal={IEEE Trans. Commun.}, 
	title={Optimal Channel Estimation for Hybrid Energy Beamforming Under Phase Shifter Impairments}, 
	year={Jun. 2019},
	volume={67},
	number={6},
	pages={4309-4325},
	doi={10.1109/TCOMM.2019.2901790}}

@ARTICLE{6459499,
	author={A. Tom and A. Sahin and H. Arslan},
	journal={IEEE Commun. Lett.}, 
	title={Mask Compliant Precoder for {OFDM} Spectrum Shaping}, 
	year={Mar. 2013},
	volume={17},
	number={3},
	pages={447-450},
	doi={10.1109/LCOMM.2013.020513.122495}}

@ARTICLE{7485853,
	author={R. Kumar and A. Tyagi},
	journal={IEEE Trans. Cognit. Commun. Netw.}, 
	title={Computationally Efficient Mask-Compliant Spectral Precoder for {OFDM} Cognitive Radio}, 
	year={Mar. 2016},
	volume={2},
	number={1},
	pages={15-23},
	doi={10.1109/TCCN.2016.2577039}}

@ARTICLE{9214877,
	author={S. Kant and M. Bengtsson and G. Fodor and B. Göransson and C. Fischione, Carlo},
	journal={IEEE Trans. Wireless Commun.}, 
	title={{EVM}-Constrained and Mask-Compliant {MIMO}-{OFDM} Spectral Precoding}, 
	year={Jan. 2021},
	volume={20},
	number={1},
	pages={590-606},
	doi={10.1109/TWC.2020.3027345}}

@article{wang2018hybrid,
	title={Hybrid precoder and combiner design with low-resolution phase shifters in {mmWave} {MIMO} systems},
	author={Z. {Wang} and M. {Li} and Q. {Liu} and A. L. {Swindlehurst}},
	journal={IEEE J. Sel. Topics Signal Process.},
	volume={12},
	pages={256--269},
	year={May 2018}
}

@article{taheri2020joint,
	title={Joint spectral-spatial precoders in {MIMO}-{OFDM} transmitters},
	author={T. Taheri and M. Mohamad and R. Nilsson and J. van de Beek},
	journal={Signal Process.},
	volume={172},
	pages={107538},
	year={Jul. 2020},
	publisher={Elsevier}
}

@ARTICLE{7913599,
		author={F. Sohrabi and W. {Yu}},
		journal={IEEE J. Sel. Areas Commun.}, 
		title={Hybrid Analog and Digital Beamforming for {mmWave} {OFDM} Large-Scale Antenna Arrays}, 
		year={Apr. 2017},
		volume={35},
		number={7},
		pages={1432-1443},
		doi={10.1109/JSAC.2017.2698958}}

@ARTICLE{7397861,
		author={X. {Yu} and J-.C. {Shen} and J. {Zhang} and K.B. Letaief},
		journal={IEEE J. Sel. Topics Signal Process.}, 
		title={Alternating Minimization Algorithms for Hybrid Precoding in Millimeter Wave {MIMO} Systems}, 
		year={Feb. 2016},
		volume={10},
		number={3},
		pages={485-500},
		doi={10.1109/JSTSP.2016.2523903}}

@ARTICLE{7366560,
	author={A. Tom and A. Şahin and H. Arslan},
	journal={IEEE Trans. Commun.}, 
	title={Suppressing Alignment: Joint {PAPR} and Out-of-Band Power Leakage Reduction for {OFDM}-Based Systems}, 
	year={Dec. 2016},
	volume={64},
	number={3},
	pages={1100-1109},
	doi={10.1109/TCOMM.2015.2512603}}

@ARTICLE{6928432,
	author={L. Liang and W. Xu and X. Dong},
	journal={IEEE Wireless Commun. Lett.}, 
	title={Low-Complexity Hybrid Precoding in Massive Multiuser {MIMO} Systems}, 
	year={Oct. 2014},
	volume={3},
	number={6},
	pages={653-656},
	doi={10.1109/LWC.2014.2363831}}

@book{bertsekas1999nonlinear,
	title={Nonlinear Programming},
	author={D. P. {Bertsekas}},
	isbn={9781886529052},
	year={3rd ed., 2016},
	publisher={Athena Scientific}
}
	\newpage

\end{document}